\documentclass[12pt, draftclsnofoot, onecolumn]{IEEEtran}
\usepackage{multicol}

\newcommand{\mc}{\mathcal}
\newcommand{\bm}{\boldsymbol}
\newcommand{\bth}{\boldsymbol{\theta}}
\newcommand{\bphi}{\boldsymbol{\phi}}

\newcommand{\pr}{\text{Pr}}
\usepackage{amsfonts,upgreek,pifont,amssymb,enumerate,color,algorithm,theorem}
\usepackage{algpseudocode}
\usepackage{cite,graphicx,epstopdf}
\usepackage{mathtools}
\usepackage{subcaption}
\newcommand{\qedsymbol}{$\blacksquare$}

\newenvironment{proof}
    {
      \emph{Proof.}
    }
    {
      \hfill\qedsymbol
    }
\newtheorem{Lemma}{Lemma}
\newtheorem{proposition}{Proposition}

\newtheorem{Corollary}{Corollary}

\newtheorem{Remark}{Remark}

\begin{document}
\title{Caching Transient Content for IoT Sensing: Multi-Agent Soft Actor-Critic}

\author{
Xiongwei~Wu, {\it Student Member, IEEE}, 
Xiuhua Li, {\it Member, IEEE}, \\Jun Li, {\it Senior Member, IEEE}, P. C. Ching, {\it Fellow, IEEE}, \\ Victor C. M. Leung, {\it  Fellow, IEEE}, H. Vincent Poor, {\it Fellow, IEEE}

\thanks{Part of this work has been submitted at the IEEE Globecom, 2020.}
\thanks{X. Wu and P. C. Ching are with the Department of Electronic Engineering, The Chinese University of Hong Kong, Shatin, Hong Kong SAR of China  (e-mail: xwwu@ee.cuhk.edu.hk; pcching@ee.cuhk.edu.hk).}
\thanks{X. Li is with the School of Big Data \& Software Engineering, Chongqing University, chongqing, 401331  China, and also with the Institute of Intelligent Network and Edge Computing at Key Laboratory of Dependable Service Computing in Cyber Physical Society (Chongqing University), Ministry of Education, China (email: lixiuhua1988@gmail.com)}
\thanks{J. Li is with the School of Electronic and Optical Engineering, Nanjing University of Science and Technology, Nanjing 210094, China. (e-mail:  jun.li@njust.edu.cn).}
\thanks{V. C. M. Leung is with College of Computer Science \& Software Engineering, Shenzhen University, Shenzhen 518060, China, and also with the Department of Electrical and Computer Engineering, The University of British Columbia, Vancouver, BC V6T 1Z4, Canada (email: vleng@ieee.org).}
\thanks{H. V. Poor is 
with the Department of Electrical Engineering, Princeton University, Princeton, NJ 08544. (e-mail: poor@princeton.edu).}
}

\maketitle
\begin{abstract}

Edge nodes (ENs) in Internet of Things commonly serve as gateways to cache sensing data while providing accessing services for data consumers. This paper considers multiple ENs that cache sensing data under the coordination of the cloud. Particularly, each EN can fetch content generated by sensors within its coverage, which can be uploaded to the cloud via fronthaul and then be delivered to other ENs beyond the communication range. However, sensing data are usually transient with time whereas frequent cache updates could lead to considerable energy consumption at sensors and fronthaul traffic loads. Therefore, we adopt age of information to evaluate data freshness and investigate intelligent caching policies to preserve data freshness while reducing cache update costs. Specifically, we model the cache update problem as a cooperative multi-agent Markov decision process with the goal of minimizing the long-term average weighted cost. To efficiently handle the exponentially large number of actions, we devise a novel reinforcement learning approach, which is a discrete multi-agent variant of soft actor-critic (SAC). Furthermore, we generalize the proposed approach into a decentralized control, where each EN can make decisions based on local observations only. Simulation results demonstrate the superior performance of the proposed SAC-based caching schemes.
\end{abstract}


\section{Introduction}
With the advancement of wireless access technology, it is envisioned that billions of devices will access the Internet, forming the so called {\it Internet of Things (IoT)} \cite{madakam2015internet}. The advent of this paradigm generalizes the accessibility towards various kinds of IoT sensors (e.g., smart cameras and temperature sensors), and thus enables intelligent services to improve human quality of life \cite{he2017software,chen2019artificial}. However, countless electronic devices are anticipated to generate a sheer volume of traffic loads, which can possibly make wireless networks saturate and degrade the quality of service. To overcome these challenges, edge nodes (ENs), e.g., small-cell base stations, are expected to act as gateways to cache sensing data close to the consumers. Consequently, it can greatly reduce traffic loads, transmission delay, and energy cost in IoT sensing networks \cite{tran2017collaborative, wang2017multi}.

Currently, some existing studies have been devoted to caching policies at wireless networks in terms of optimizing communication performance criteria, e.g., traffic loads, latency, and power consumption \cite{wu2020joint,wu2019jointMDS,li2018hierarchical,li2015distributed,li2016pricing}. These caching policies emphasize how to efficiently cache multimedia content given limited storage at ENs. However, IoT sensors usually generate sensing data at a relatively small size  \cite{xu2020aoi}. Therefore, each EN can be assumed to have enough storage to cache content items produced by all sensors in the network \cite{xu2020aoi}. In this way, each EN can locally satisfy user requests towards all content items.
Moreover, in contrast with multimedia content that is often {\it in-transient}, sensing data cached at ENs gradually become outdated as time passes. The staleness of caching content may significantly deteriorate the performance of IoT sensing services. Indeed, how to preserve data freshness constitutes the primary challenge in designing caching policies for IoT sensing. 
A recently proposed performance criterion can be adopted to quantify data freshness, namely, {\it age of information} (AoI) \cite{kaul2012real}. The AoI of a content item is defined as the amount of time that has passed since the last measurement of this content item. Given the arrivals of user requests, cache update is needed to reduce the average AoI of caching content items \cite{kaul2012real}. Nevertheless, excessive cache updates will generate considerable energy consumption and challenge the battery life of sensors. Hence, these characteristics of IoT sensing require new and efficient caching polices in IoT sensing networks.

\subsection{Related Work}

AoI was initially investigated in \cite{kaul2012real} to evaluate status update for packet delivery between a source node and destination node. The author derived the average AoI by considering a simple queuing model. Such a source-destination scenario was further investigated in \cite{huang2015optimizing
} under a more complex queuing system, i.e., $M/G/1$. The study in \cite{costa2016age} also focused on this scenario and investigated average and peak AoI. The authors in \cite{sun2017update} studied the optimal policy for packet delivery from a source to a remote destination by considering an age penalty function.  
In general, these works were extensions of the study in \cite{kaul2012real}, which characterized the average or peak AoI based on different types of queuing models. Later on, AoI was adopted to evaluate the performance of IoT sensing networks, which involved multiple sensors and error-prone wireless links. The study in \cite{emara2020spatiotemporal} investigated the peak AoI under different scales of IoT networks by considering Bernoulli traffic. Taking into account sampling cost and update cost, the research in \cite{zhou2019joint} examined update policies under a single sensor scenario and multiple sensors scenario, respectively. Moreover, the studies in \cite{gong2018energy,gu2019timely} investigated the age-energy tradeoff, and characterized AoI and energy cost in closed form. This line of research generally attempted to derive analytical expressions of age-based performance metrics under queuing models, and focused on performance analysis by utilizing  optimization theory. 

Some recent works have investigated intelligent policies for cost-effective caching in IoT sensing networks by applying reinforcement learning (RL). The authors in \cite{ceran2018reinforcement} considered a single sensor and proposed to minimize the average AoI subject to the average number of updates. Treating the model parameters of neural networks as transient content, the study in \cite{ma2020deep} proposed to minimize the average AoI plus cost by using deep Q-network (DQN). Both studies evaluated cache update cost by counting the number of content transmissions. The study in \cite{ceran2019reinforcement} proposed to minimize the average AoI by considering sensing and transmission energy cost. The authors in \cite{hatami2020age} considered cache update at multiple sensors, and investigated the tradeoff between energy consumption and AoI via Q-learning. Similarly, the tradeoff between AoI and energy consumption was also investigated in studies\cite{xu2020aoi,abd2020reinforcement}. Moreover, the studies in \cite{zhu2018cachingTransient,yao2020caching} utilized actor-critic (AC)-based approaches to investigate how to cache transient content by considering content update cost. All of these studies focused on cost-effective update polices at a single EN, 
which aggregates sensing data generated from sensors at its coverage.  

\subsection{Contributions}
This paper investigates intelligent policies of cost-effective cache update in the IoT sensing network, where multiple ENs cache sensing data under the coordination of the cloud. Compared with prior studies on a single EN that only entails sensors at its coverage \cite{ceran2018reinforcement,ceran2019reinforcement,hatami2020age,abd2020reinforcement,xu2020aoi,zhu2018cachingTransient,yao2020caching}, this work investigates a more general scenario. Specifically, each EN is likely to communicate with a subset of sensors because of the short communication ranges at IoT sensors \cite{hassanalieragh2015health}. Thus, each EN needs to upload content items generated from its coverage to the cloud so that other ENs beyond the communication range can download these content items and provide accessing services for data consumers. Consequently, cache update at multiple ENs not only requires energy consumption at sensors but also leads to fronthaul traffic loads. Given limited battery levels at sensors and capacity of fronthaul links in reality, it is imperative to find cache update policies that preserve data freshness while reducing update costs. For this reason, we consider a more integrative performance metric, involving the average AoI, energy consumption and fronthaul traffic loads. In addition, the average AoI in a multiple ENs scenario is characterized by user requests received at all ENs. We need to take into account space-time dynamics of content popularity \cite{sadeghi2018optimal} in comparison to the studies \cite{ceran2018reinforcement,ceran2019reinforcement,hatami2020age,abd2020reinforcement,xu2020aoi}. 

The considered scenario results in a multi-agent discrete decision-making, where the space of the discrete decisions grows exponentially versus system parameters, e.g., the number of ENs. Some conventional RL algorithms used in prior studies, e.g., DQN, suffer from high sample complexity\footnote{Sample complexity of a RL algorithm usually refers to the number of training experiences that an agent needs to generate in order to achieve certain level of reward.} and brittleness to scalability, which are not efficient in handling the multi-agent tasks \cite{zhang2019multi}. 
We therefore devise a novel RL approach with an output size linearly increases w.r.t. these system parameters. The proposed approach utilizes the idea of the state-of-the-art RL algorithm, similar to the soft actor-critic (SAC) in \cite{haarnoja2018soft}, that is originally applicable to continuous decision-making only.





The main contributions of this paper are summarized as follows.
	\begin{itemize}
		\item 
		To the best of our knowledge, this work is the first to investigate the issue of cache update at multiple ENs in IoT sensing. 
		We formulate a cache update problem as a cooperative multi-agent MDP, with the goal of minimizing the average AoI of caching content items plus cache update costs, i.e., transmission energy consumption and fronthaul traffic loads.

		\item 
		To deal with the formulated problem, 
		we devise a multi-agent discrete variant of SAC with an output size that linearly increases versus the numbers of ENs and IoT sensors.
		The core idea is that we customize the {\it Gumbel-SoftMax} (GS)-sampler to approximately generate differentiable actions. 
		Meanwhile, the utilization of entropy regularization can effectively enhance {\it exploration}, which assists to prevent premature convergence. 
		%

		\item 
		To reduce the communication overhead between ENs and the cloud, 
		we further generalize the proposed centralized algorithm into decentralized control. 
		Particularly, each EN serves as an independent agent with a decentralized policy, exploring its caching decisions based on local observations. Moreover, we maintain the centralized soft Q-function at the cloud processor (CP) to augment EN coordination, favorably improving system reward.

		\item Simulation results are presented to demonstrate that the proposed RL approach outperforms existing RL-based caching policies, and unveil how transmission energy and fronthaul traffic load considerations compromise data freshness in IoT sensing networks. 
\end{itemize}

The remainder of this paper is organized as follows. Section II introduces the system model. Section III describes the problem formulation. Section IV develops a centralized DRL-based cache update scheme, and Section V develops a decentralized DRL-based cache update scheme. Section VI shows the performance evaluation, and Section VII concludes the paper. 

\section{System Model}
As depicted in Fig. \ref{fig:system_iot}, consider an IoT sensing network, in which a total of $B$ ENs are connected to the CP through wired fronthaul links. Every EN is equipped with a cache unit and a computing unit, which empower edge caching and edge computing, respectively. Consequently, such wireless networks allow ENs to serve as gateways between IoT sensors and data consumers \cite{zhu2018cachingTransient}. More specifically, ENs are capable of caching sensing data generated by various kinds of sensors within their communication ranges. Meanwhile, data consumers can submit their requests to ENs and retrieve corresponding information for data processing and analysis. Generally, data consumers in IoT applications are able to make requests from static and mobile devices (e.g., computers, phones, vehicles,  etc.) \cite{chen2019artificial}. For instance, data consumers can inspect temperatures or humidities of the environment on mobile applications.  
For ease of discussion, we assume that, each EN coordinates $F$ sensors that are randomly distributed within its coverage. As such, there are a total of $B\times F$ sensors in the network. Each sensor is supposed to communicate with the nearest EN. 
Let $\mc B = \{1,2,\cdots, B\}$ and $\mc F = \{1,2, \cdots, B\times F\}$ denote the indices of ENs and sensors. Moreover,  $\mc F_b = \{(b-1)F+1, \cdots, bF\}$ denotes the indices of the sensors coordinated by the $b$-th EN. That is, $\mc F = \cup_{b \in \mc B} \mc F_b$. 
\begin{figure}[t]
  \centering
  \includegraphics[scale=.95]{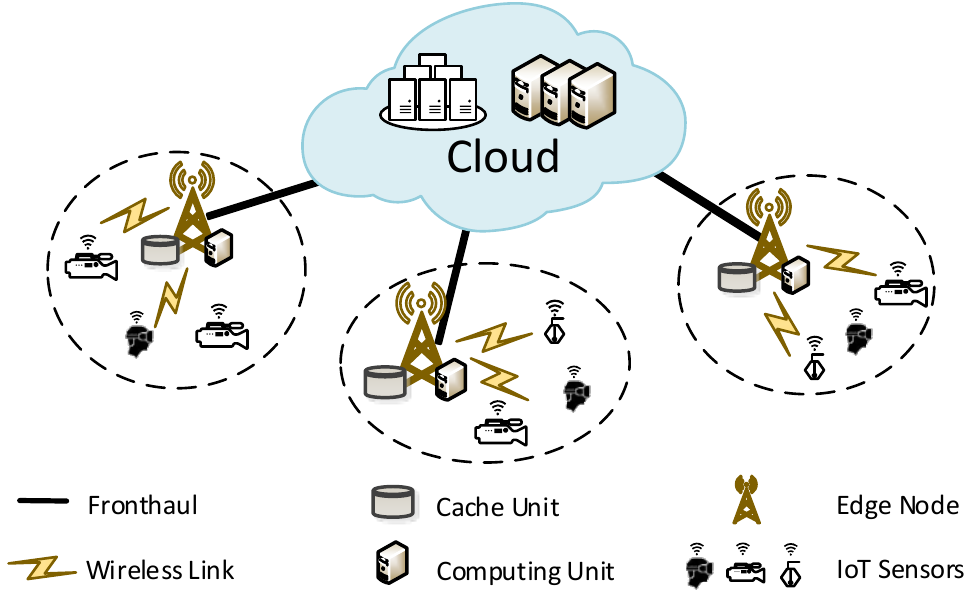}
  \caption{Illustration of an IoT sensing network.}
  \label{fig:system_iot}
\end{figure}

\subsection{Age of Information}
The system operation is assumed to be slotted into a sequence of discrete epochs, i.e., $t = 1, 2, \cdots$. In IoT sensing, each content item cached at the EN, is generated by certain IoT sensor. For instance, content item\footnote{
By slightly abusing the notation, we denote the index of either a content item or an IoT sensor by $f$, and $f \in \mc F$.}  $f \in \mc F$ implies that this content item is produced by the $f$-th sensor. In general, each caching content item can be temporally updated and replaced by a new version of the sensing data.
As such, we denote the generation epoch for the version of content item $f$ cached at epoch $t$ by $v_f^t$. Evidently, $v_f^t \leq t$. 
To evaluate data freshness of a content item, we adopt AoI as the QoS metric, which counts how many epochs has passed since this content item was produced. 
In this way, the AoI of a caching content item $f$ can be calculated as follows: 
\begin{align}
	o_f^t = \min\{t - v_f^t, 1\}, \forall f \in \mc F,
\end{align}
which takes value from a finite range, i.e., $\{ 1, 2, \cdots T_{\max}\}$; and $T_{\max}$ denotes the upper limit, which implies the most outdated level of a content item \cite{abd2020reinforcement,gindullina2020age,xu2020aoi}. We consider that each EN is able to receive user requests concerning content items generated by all of the sensors (e.g., $\forall f \in \mc F$). This is reasonable in real applications because data consumers usually have diverse preferences towards content items. 
Let $\{N_{f,b}^t\}_{f\in \mc F, b \in \mc B}$ be the number of user requests received by ENs at epoch $t$. 
Consequently, the average AoI to satisfy user demands at epoch $t$ can be calculated as follows \cite{ma2020deep}:
\begin{align}
	O^t = \frac{\sum_{f \in \mc F, b \in \mc B} o_f^t N_{f,b}^t} {\sum_{f \in \mc F, b \in \mc B} N_{f,b}^t}. \label{eq:avgaoi}
\end{align}
As aforementioned, IoT sensing data are transient and gradually become stale as time passes. 
Therefore, it is necessary to temporally renew caching content items, so as to maintain favorable data freshness. As shown in  Fig. \ref{fig:aoidemo}, the AoI of a transient content item increments by one after every epoch; once this content item is selected to update, e.g., at epoch $t_1$ or $t_2$, the corresponding  AoI reduces to 1 at the subsequent epoch.  

\subsection{Cache Update}
To perform cache update, ENs should communicate with sensors through wireless links. 
Owing to channel fading, we assume the following successful transmission condition: data transmissions between IoT sensors and ENs are successful only on condition that the received SNR exceeds a pre-defined threshold $\eta_{th}$. Specifically, we assume that orthogonal channels are scheduled to different sensors. Thus, the received SNR for sensor $f$ delivering a content item to the associated EN can be expressed as:
\begin{align}
	\eta_f = \frac{P_{f} \chi_f^2 \kappa_f^2}{N_0B_0}, \forall f \in \mc F, 
\end{align}
where $P_f$ is transmission power at sensor $f$; coefficient $\chi_f$ denotes the large-scale fading; $N_0$ denotes noise power spectrum density; and $B_0$ is the channel bandwidth. In addition, $\kappa_f$ denotes the envelope of the small-scale fading, which is assumed to follow the Rayleigh distribution \cite{lee1990estimate}, e.g., $ \mathbb{P}_{\kappa_f} (\kappa_f) =  \kappa_f \exp(- \kappa_f^2/2)$.  Let $s_f$ be the storage of content item $f$.  
Subsequently, the average transmission energy consumption for cache update, determined by channel gain and content size, is characterized as follows. 
\begin{proposition}
	\label{prop:egy}
	The average transmission energy $\bar E_f$ at the $f$-th IoT sensor ($\forall f \in \mc F)$ for dispatching sensing data to the CP is as follows:
	\begin{align}
		 \bar E_f = \frac{\log2 \times P_f s_f}{\log2 \times R_{th} \exp\left(-\frac{\eta_{th}}{2\beta_f}\right)  + B\exp\left(\frac{1}{2\beta_f}\right)\mc \rho_f(\eta_{th} +1)}, \label{eq:avgegy}
	\end{align}
	where function $\rho_f(\cdot)$ is defined as:
	\begin{align}
	 	\rho_f(x) \triangleq \int_{x}^{\infty} \frac{1}{x} \exp(-{x}/{(2 \beta_f)}) dx, \label{eq:rpo}
	 \end{align} 
	 and $\beta_f = P_f \chi_f^2 / (N_0B)$; and $R_{th}$ denotes the throughput threshold, i.e.:
	 \begin{align}
	 	R_{th} \triangleq \log_2(1 + \eta_{th}).
	 \end{align}
\end{proposition}
\begin{proof}
  See Appendix \ref{appen:A}.
\end{proof}

\begin{figure}[t]
  \centering
  \includegraphics[scale=1.9]{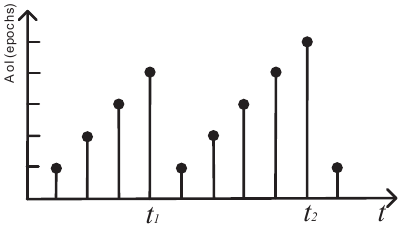}
  \caption{AoI evolution. Cache update happens at $t_1$ and $t_2$.}
  \label{fig:aoidemo}
\end{figure}

When a cached content item is updated, the associated EN should deliver the updated content item to the CP via fronthaul. Thus, if other ENs overhear requests that relate to sensors out of their coverage, they can fetch these content items from the CP. However, frequent fronthaul transmissions will impose tremendous traffic loads. We therefore allow each EN to pro-actively cache content items that are generated beyond their communication range and kept at the CP. Specifically,  when a new version of content item $f \in \mc F_b$ is uploaded to the CP, other ENs (i.e., $\forall b' \neq b$) should fetch this content item via fronthaul transmissions with traffic loads $(B-1) s_f$. We consider that the storages of cache units in ENs are sufficiently large enough to aggregate content items generated by IoT sensors (e.g., $\forall f \in \mc F$) in the network because the storage of sensing data is often at a small size in practice. 

As previously stated, since battery levels of sensors and capacity of fronthaul links are restricted in reality, caching content should be reasonably updated to achieve a favorable tradeoff among average AoI, energy consumption, and fronthaul traffic load. In view of this, we formulate a cache update decision-making strategy in the next section.

\section{MDP Problem Formulation}
Our goal is to find cache update policies, which allow ENs to reasonably update content items under different states, minimizing the long-term average weighted cost. This weighed cost is supposed to comprise average AoI, transmission energy, and fronthaul traffic loads. 
\subsection{Multi-Agent Cooperative MDP}
Note that, the average AoI  (e.g., \eqref{eq:avgaoi}) critically depends on content popularity at  ENs:
\begin{align}
	\hat p_{f,b} = N_{f,b}/\textstyle \sum_{f' \in \mc F, b \in \mc B} N_{f',b}, \forall f, b, 
\end{align}
which is usually time-varying in practice \cite{sadeghi2018optimal}. Furthermore, transmission energy consumption (e.g., see \eqref{eq:avgegy}) entails storages of content items and statistics of wireless channels that are often inhomogeneous towards distinct sensors. 
To efficiently coordinate ENs to make cache update decisions in such a complex environment, we formulate the cache update decision-making as a multi-agent cooperative MDP. Particularly, every EN is anticipated to play a role of an agent, and we define the basic elements of a multi-agent MDP as follows. 

\begin{itemize}
	\item {\bf State:} we denote state space by $\mc S$, which contains all possible states $\bm s$. Every state $\bm s$ consists of local observations of all agents, i.e., $\bm s = \{ \bm s_1, \bm s_2, \cdots, \bm s_B\}$.  In the IoT sensing network, each agent (i.e., EN) is capable of observing the AoI of every content item and local user requests, i.e.,
	\begin{align}
		\bm s_b^t =  \left(\{o_f^t\}_{f \in \mc F}, \{N_{f,b}^t\}_{f\in \mc F}\right), 
		\forall b \in \mc B,
	\end{align}
	which is a ${2BF}$-dimensional tuple.

	\item {\bf Action:} Let $\bm a = \{ \bm a_1, \bm a_2, \cdots, \bm a_B\}$ denote a joint action, where local action $\bm a_b$ implies which content item should be selected to update. Similar to \cite{hatami2020age}, each EN $b$ is assumed to select at most one content from $\mc F_b$ at each epoch\footnote{The proposed framework can also be generalized to case where multiple content items are determined to update at each EN. This is, however, at the cost of  larger system bandwidth and energy consumption.}. Accordingly, the local action space of agent $b$ can be given by $\mc A_b = \{0\} \cup \mc F_b$; and the joint action space is given by:
	\begin{align}
		\mc A = \cup_{b \in \mc B} \mc A_b.
	\end{align}
	 Particularly, when local action $\bm a_b = 0$, it implies that EN $b$ remains idle and presents null transmission energy and traffic loads. Otherwise, the corresponding sensor needs to upload the current measurement of content item $\bm a^t \in \mc F$ into EN $b$. That is, 
	\begin{align}
		o_f^{t+1} = \min \big\{ (o_f^t + 1) \times (1- \mc I(f , \bm a_b^t)) + \mc I (f, \bm a_b^t) , T_{\max}\big\}, \forall f \in \mc F_b, \forall b \in \mc B,  
	\end{align}
	where $\mc I(\cdot)$ is an indicator function\footnote{Given parameters $x,y$, when $x=y$, we have $\mc I(x,y) = 1$; otherwise, $\mc I(x, y) = 0$.}. After each agent takes action $\bm a_b^t$, the system state becomes $\bm s^{t+1}$ with transition probability $\text{Pr}\{ \bm s^{t+1}|\bm s^t, \bm a^t\}$ at epoch $t+1$.

	\item {\bf Reward:}  In a multi-agent cooperative MDP, all agents are expected to share a common reward $r^{t+1}$, which unveils how effective a joint action $\bm a^t$ is \cite{zhang2019multi}. Recall that our objective is to minimize the average AoI  whilst reducing transmission energy consumption and fronthaul traffic loads. Consequently, we define the average weighted cost at each epoch as follows:
		\begin{align}
		C^{t+1} = \frac{\sum_{f \in \mc F, b \in \mc B} o_f^{t+1} N_{f,b}^{t+1}} {\sum_{f \in \mc F, b \in \mc B} N_{f,b}^{t+1}} + \omega_1 \sum_{b \in \mc B} \bar E_f|_{f = \bm a_b^t} +  \omega_2 \sum_{b \in \mc B} (B-1) s_f|_{f = \bm a_b^t}, \label{eq:cost}
	\end{align}
	where the first term on the right-hand side is the average AoI to satisfy user demands arrived at epoch $t+1$; $\omega_1$ and $\omega_2$ are non-negative coefficients to weigh the importance of energy and traffic cost. For notational simplicity, we define $\bar E_0 = 0$ and $ s_0 = 0$, respectively. In accordance with reward maximization, we define the reward $r^{t+1} \triangleq R(\bm s^{t+1}, \bm s^t, \bm a^t) $, where 
	\begin{align}
		R(\bm s^{t+1}, \bm s^t, \bm a^t) = - C^{t+1}, \label{eq:insreward}
	\end{align}
	which is a negative value. 
\end{itemize}
Consequently, we aim to find a caching policy $\pi^*$, which is able to generate a joint action $\bm a$ given any state $\bm s$ that maximizes the expected discounted cumulative reward as follows:
\begin{align}
	\pi^* = \arg\max_{\pi} \mathbb{E} [V^t|\pi], \label{eq:rl} 
\end{align}
where 
\begin{align}
	V^t = \sum_{\tau = 0}^{\infty} \gamma^{(\tau)} r^{t+\tau+1}, \label{eq:v}
\end{align}
and $\gamma \in [0,1)$ is a discounted factor. 

To address problem \eqref{eq:rl}, one can resort to model-based approaches \cite{sutton2018reinforcement}, which usually rely on the knowledge of Pr$\{\bm s^{t+1}|\bm s^t, \bm a^t\}$. However, in practice, transition probability is usually uncertain and difficult to estimate. Even if we could have this knowledge, our problem is still intractable due to the curse of dimensionality in multi-agent settings. 
These challenges motivate us to explore data-driven approaches, i.e., RL, which is as a consequence of properly utilizing past experiences. In the following sections, we develop efficient RL algorithms to handle the formulated problem. 

\section{Centralized Multi-Agent Discrete Soft Actor-Critic-Based Caching}

In this section, we develop a centralized RL algorithm for cache update in IoT sensing network, where the CP acts as the centralized agent and coordinates caching decisions for all ENs. To what follows, we first outline the background of RL and identify the challenges of conventional RL algorithms to solve the considered problem. Then, we devise an efficient RL approach, which is a multi-agent discrete variant of the state-of-the-art RL. Finally, we present the concrete algorithm implementation.  

\subsection{Background of Reinforcement Learning}
Canonical RL generally aims to estimate the following Q-function: 
\begin{align}
	Q^*(\bm s, \bm a) = \mathbb[V^t|\bm s^t = \bm s, \bm a^t = \bm a, \pi^*], 
\end{align}
which indicates the expected cumulative reward after taking action $\bm a^t$ under state $\bm s^t$,  subsequently following policy $\pi^*$. An optimal policy can be characterized by the following {\it Bellman Optimality}.
\begin{Lemma}
\label{eq:bellop}
	An optimal policy $\pi^*$ leads to the following recursive equations \cite{sutton2018reinforcement}:
\begin{align}
	Q^*(\bm s, \bm a) = R(\bm s', \bm s, \bm a)  + \gamma \max_{\bm a' \in \mc A} Q^*(\bm s', \bm a'),
\end{align}
where $Q^*(\bm s, \bm a)$ denotes the optimal Q-function by following $\pi^*$.
\end{Lemma}

Lemma \ref{eq:bellop} lays the foundation for DQN. As a popular approach for discrete decision-marking, DQN utilizes deep neural networks (DNNs) as function approximators to predict the optimal Q-function.
Readers are referred to \cite{mnih2015human} in detail. Consequently, an optimized policy can be given by a mapping as follows:
\begin{align}
	\bm s^{t} \rightarrow  \arg\max_{\bm a \in \mc A} \widehat Q (\bm s^t, \bm a).
\end{align}
In other words, DQN needs to output the value of Q-function over all possible discrete actions given any state (i.e., 
$
	\bm s \rightarrow \mathbb{R}^{|\mc A|}
$). This practice results in slow convergence and brittleness to scalability. 

For this reason, DQN is extremely difficult to be applied in multi-agent and high dimensional settings \cite{haarnoja2018soft}. In the considered problem, the size of action space exponentially increases versus the number of ENs and polynomially increases versus the number of sensors, i.e., $(F+1)^B$. 
For instance, consider a simple setting: three ENs are deployed, each of which coordinates $F = 10$ sensors; the number of resulting discrete actions is 1331, which leads to large network size and slow convergence. When $F = 20$, the number of discrete actions goes up to $9261$ that is almost intractable. To overcome these challenges, we can adopt the AC-based approaches, where an independent function approximator is utilized to generate actions, instead of relying on the Q-function. 

\subsection{Proposed Multi-Agent Discrete Soft Actor-Critic Learning}
SAC is the state-of-the-art RL algorithm, which is as a result of an entropy regularized formalism that augments {\it exploration} \cite{haarnoja2018soft}. 
This approach entails an AC framework, which specifies stochastic policy and soft Q-function separately. That is, SAC attempts to find a stochastic policy that maximizes the expected cumulative reward while taking diverse actions as many as possible. Consequently, SAC is able to achieve high sample efficiency \cite{haarnoja2018soft}. However, the SAC in \cite{haarnoja2018soft} is only applicable in continuous settings. We now develop a multi-agent discrete variant of SAC learning, which is suitable to handle discrete decision-making especially in high dimensional settings. 

Similarly, our objective is to find a stochastic policy $\pi(\bm a|\bm s)$ that maximizes the expected cumulative reward plus its entropy, i.e.:
\begin{align}
	\pi^* = \arg\max_{\pi} \mathbb{E}_{\{\bm s^t, \bm a^t\}} \left[ \sum_{t=0}^{+\infty}(\gamma)^{t}  \bigg(r^{t+1} + \alpha \mc H(\pi(\cdot|\bm s^t)) \bigg) \right], \label{eq:sac} 
\end{align}
where $\pi(\cdot |\bm s^t)$ is a categorical distribution indicating the probability of taking any action under state $\bm s^t$; $\mc H(\cdot)$ denotes the entropy of a distribution; $\alpha$ is the temperature parameter and controls the magnitude of entropy regularization. In general, a lager $\alpha$ prompts the agents to carry out a more random {\it exploration} in making decisions. 

Accordingly, the soft Q-function can be defined as follows:
\begin{align}
	Q(\bm s, \bm a) = \mathbb{E} \left[ V^t +  \alpha \sum_{\tau=1}^{+\infty}(\gamma)^{\tau}  \mc H(\pi(\cdot|\bm s^{t+\tau}))  \big|\bm s^t = \bm s, \bm a^t = \bm a \right],
\end{align}
where $V^t$ is the discounted cumulative reward, given by \eqref{eq:v}. This further gives rise to the soft value function:
\begin{align}
	V(\bm s) = \mathbb E_{\bm a \sim \pi } \left[ Q(\bm s, \bm a) - \alpha \log \pi (\bm a|\bm s) \right]. \label{eq:sv}
\end{align}
According to Lemma \ref{eq:bellop}, we have the similar recursive equality as follows:
\begin{align}
	Q(\bm s^{t}, \bm a^t) = \mathbb{E}_{\bm s^{t+1}} \left [ r^{t+1} + \gamma V(\bm s^{t+1}) \right].  \label{eq:sq} 
\end{align}
To pave the way for the multi-agent discrete SAC (MADSAC), we introduce the {\it Soft Policy Improvement} \cite{haarnoja2018soft}.
\begin{Lemma}
Given policy $\pi_{old}$ and soft Q-function $Q_{old}(\bm s, \bm a)$ with a finite size of action space $\mc A$, a new policy $\pi_{new} $ can be calculated by:
\begin{align}
 	\min_{\pi'} D_{KL} \bigg( \pi'(\cdot |\bm s) \|\exp(Q_{old} (\bm s, \cdot)/\alpha)/Z(\bm s) \bigg), \forall \bm s \in \mc S, \label{eq:spi}
 \end{align} 
 where $D_{KL}(\cdot \|\cdot)$ denotes the operator of the Kullback-Leibler divergence, and $Z(\cdot)$ is used for normalization. Then, it leads to $Q_{new} (\bm s, \bm a) \geq Q_{old} (\bm s, \bm a)$ for any $(\bm s, \bm a) \in \mc S \times \mc A$ \cite{haarnoja2018soft,haarnoja12018soft}. 
 \end{Lemma} 

Following the elementary steps of the SAC in \cite{haarnoja2018soft}, we consider a parameterized $Q_{\bphi}(\bm s, \bm  a)$ and policy $\pi_{\bth}(\cdot| \bm s)$, where $\bphi$ and $\bth$ are parameters of some function approximators, e,g., DNNs. Since the soft value function can be expressed by \eqref{eq:sv}, we do not incorporate a separate function approximator here. In addition, $Q_{\bphi}(\bm s, \bm a)$ denotes a mapping with unit output, i.e., $(\bm s, \bm a) \rightarrow \mathbb{R}$, instead of the number of all possible discrete actions. This is because we can solely use policy function to generate decisions. However, the utilization of stochastic policy $\pi_{\bth}(\cdot| \bm s)$ in SAC still involves estimates probability distribution over all actions that increase exponentially w.r.t. the number of agents. 

To confine the output of $\pi_{\bth}(\cdot| \bm s)$ in multi-agent RL, we first propose to impose the following decomposition:
\begin{align}
	 \pi_{\bth} (\bm a |\bm s) = \Pi_{b \in \mc B} ~ \mu_{\bth} (\bm a_b|\bm s), \forall \bm a \in \mc A, \bm s \in \mc S, \label{eq:pmulti}
\end{align}
where $\bm a = \{\bm a_1, \bm a_2, \cdots, \bm a_B\}$. 
In this way, policy $\pi_{\bth}$ can be implicitly represented by a function approximator $\mu_{\bth}(\cdot|\bm s)$, which outputs only a vector of the probability of taking each local action (e.g., $\bm a_b \in \mc A_b, \forall b \in \mc B$) with output dimension $\sum_{b \in \mc B} |\mc A_b|$. That is,
$
	\sum_{\bm a_b \in \mc A_b} \mu_{\bth} (\bm a_b |\bm s) = 1,
$
for $\forall b \in \mc B$. 

Then, on the basis of \eqref{eq:sv}-\eqref{eq:spi}, we train the above parameterized functions by using historical experiences, e.g., $\xi^t = (\bm s^t, \bm a^t, r^{t+1}, \bm s^{t+1})$. We use the {\it Replay Buffer (RB)} to store some recent experiences, i.e., $ \xi^t \in \Xi$. 
Specifically, we can train policy function according to {\it Soft Policy Improvement}. 
By omitting the normalization factor $Z(\cdot)$ in \eqref{eq:spi}, the policy parameter (i.e., $\bth$) can be trained by adopting stochastic gradient descent to minimize the following loss: 
 \begin{align}
 	J_\pi (\bth) = \mathbb E_{\bm s^t \sim \Xi} \left [  \mathbb E_{\bm a \sim \pi_{\bth}} [\alpha \log(\pi_{\bth} (\bm a| \bm s^t))  - Q_{\bphi} (\bm s^t, \bm a)]\right], \label{eq:lossp}
 \end{align}
 where the expectation over $\bm s^t$ can be approximated by drawing samples from the {\it RB}. However, \eqref{eq:lossp} incorporates an expectation over actions following policy distribution $\pi_{\bth} (\cdot|\bm s^t)$. A challenging issue is that the gradient w.r.t. $\bth$ can not be backpropagated in a normal manner if we directly utilize $\pi_{\bth}(\cdot|\bm s^t)$ to generate samples. 
To deal with it, we resort to the technique of {\it reparameterization trick} \cite{jang2016categorical}. The core idea is to find samples that incorporate the parameter of  policy distribution, e.g., $\bth$.

Given a categorical distribution $\pi_{\bth}$ with decomposition in \eqref{eq:pmulti}, we have the following Lemma. 
\begin{Lemma}
\label{lemma:gs}
Suppose that 
$\{g_{i,b}, \forall i \in \mc A_b, b \in \mc B\}$ are i.i.d. and follow\footnote{That is, $x \sim \exp(- (x + \exp(-x)))$} Gumbel(0,1). Given a stochastic policy $\pi_{\bth}$ for a multi-agent discrete decision-making, a joint action $\hat {\bm a} = \{\hat {\bm a}_1, \hat {\bm a}_2, \cdots, \hat {\bm a}_B\}$ can be generated as follows:  
\begin{align}
	\hat{\bm a}_b = \arg\max_{i \in \mc A_b}~ [g_{i,b} + \log \mu_{\bth} (i| \bm s)] , \forall b \in \mc B.
\end{align}
That is, Pr$\{\bm a = \hat{\bm a}|\bm s\} = \pi_{\bth} (\hat{\bm a}|\bm s)$. 
\end{Lemma} 
\begin{proof}
	See Appendix \ref{appen:B}.
\end{proof}

The above Lemma is an extension of the results in \cite{jang2016categorical}. As a result, sample $\hat{\bm a}$ is a function of $\bth$. To make samples differentiable w.r.t. $\bth$, the non-differentiable operator $\arg\max$ can be further approximated by {\it SoftMax}, resulting in a GS-sampler. 
\begin{Corollary}
\label{coro:1}
 Suppose vector $\bm z_b = [z_{i,b}] \in \mathbb{R}^{|\mc A_b|}, \forall b \in \mc B$, where each element is given by:
 	\begin{align}
 		 z_{i,b} = \frac{ \exp((\log(\mu_{\bth} (i |\bm s)) + g_{i,b} )/c_0 ) }{\sum_{j \in \mc A_b}  \exp((\log(\mu_{\bth} (j |\bm s)) + g_{j,b} )/c_0 ) }, \forall i \in \mc A_b, \forall b \in \mc B. \label{eq:gssampler}
 	\end{align}
 Then, when coefficient $c_0$ goes to 0, $\bm z_b$ approaches an one-hot vector with one element being 1 and all other elements being 0, $\forall b \in \mc B$. 
\end{Corollary}
In other words, through reshaping local actions as one-hot vectors, a joint action can be approximately generated by using GS-sampler, i.e., $GS(\bm \pi_{\bth} (\cdot|\bm s))$, that is differentiable w.r.t. $\bth$. Toward this end, with the aid of GS-sampler, \eqref{eq:lossp} can be minimized by stochastic gradient descent, which will be detailed in the next subsection.

Regarding the updates of soft Q-function and temperature parameter $\alpha$, the procedure exactly follows the steps in \cite{haarnoja2018soft}. Specifically, the parameterized soft Q-function can be updated by minimizing the following loss term:  
\begin{align}
	J_Q(\bphi) = \mathbb{E}_{\xi^t \sim \Xi} \left[ \left( Q_{\bphi}(\bm s^t, \bm a^t) -  (r^{t+1} + \gamma \bar V (\bm s^{t+1}) )\right) ^2 \right], \label{eq:lossq}
 \end{align}
where the target value $\bar V(\bm s^{t+1})$ is given by:
\begin{align}
	\mathbb E_{\bm a \sim \pi_{\bth} } \left[ Q_{\bphi^-} (\bm s^{t+1}, \bm a) - \alpha \log \pi_{\bth} (\bm a|\bm s^{t+1}) \right], \label{eq:targetq}
\end{align}
and $Q_{\bphi^-} (\bm s, \bm a)$ is the target network with parameter $\bphi^-$ being updated in the {\it soft copy} manner \cite{Lillicrap2015ContinuousCW}. This procedure is detailed in the next subsection. 
The temperature parameter $\alpha$ is updated by minimizing the follow term:
\begin{align}
	J(\alpha) = \mathbb{E}_{\bm s^t \sim \Xi} \left[\mathbb{E}_{\bm a \sim \pi} [-\alpha (\log(\pi_{\bth} (\bm a|\bm s^t)) - \bar H)]\right], \label{eq:losstem}
\end{align}
where $\bar H$ is constant and denotes the target entropy. Notably, \eqref{eq:targetq} and \eqref{eq:losstem} again involve taking an expectation over action $\bm a$. Since action samples do not contribute to the gradient of $\bphi$ or $\alpha$, there is no need in principle to draw action samples via GS-sampler. 
\begin{figure}
\begin{subfigure}[t]{.495\linewidth}
  \centering
  \includegraphics[scale=0.7]{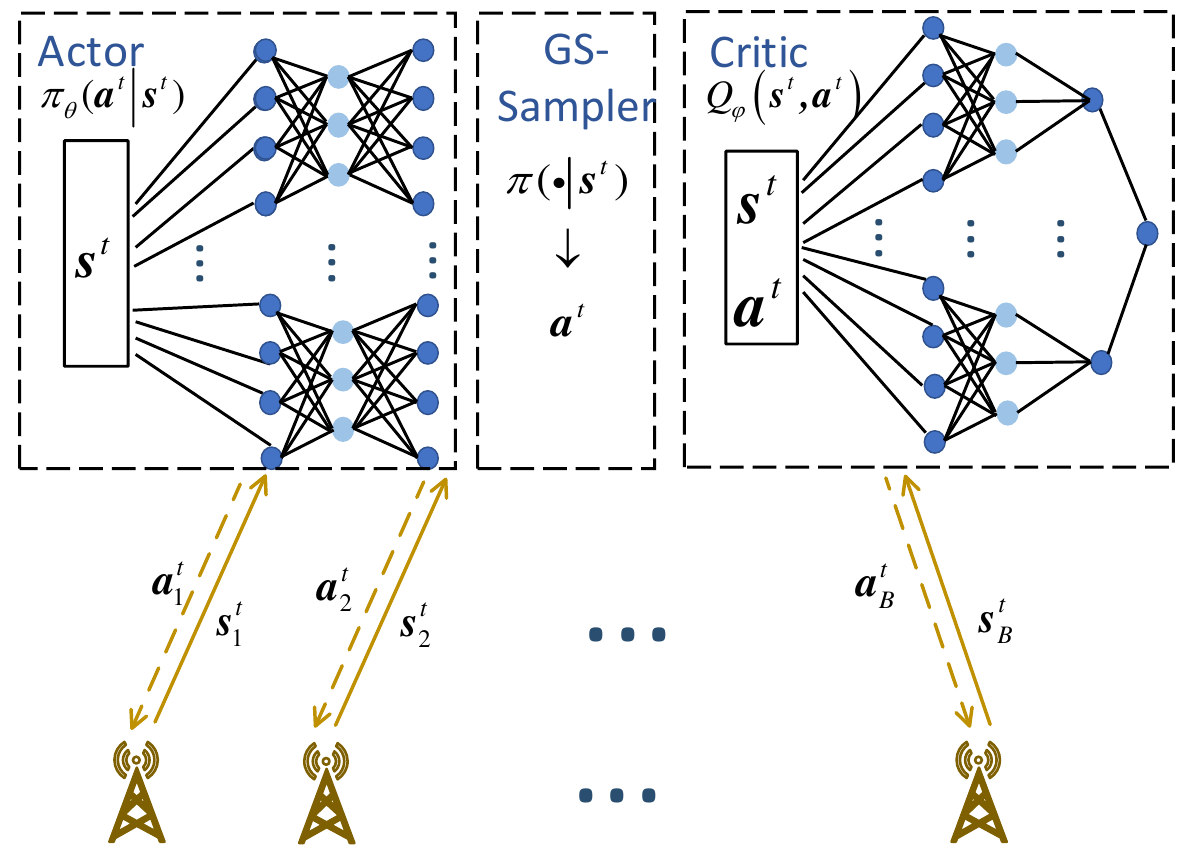}
  \caption{Proposed centralized control. }
  \label{diagram:madsac_cc}
\end{subfigure}
\begin{subfigure}[t]{.495\linewidth}
  \centering
  \includegraphics[scale=0.7]{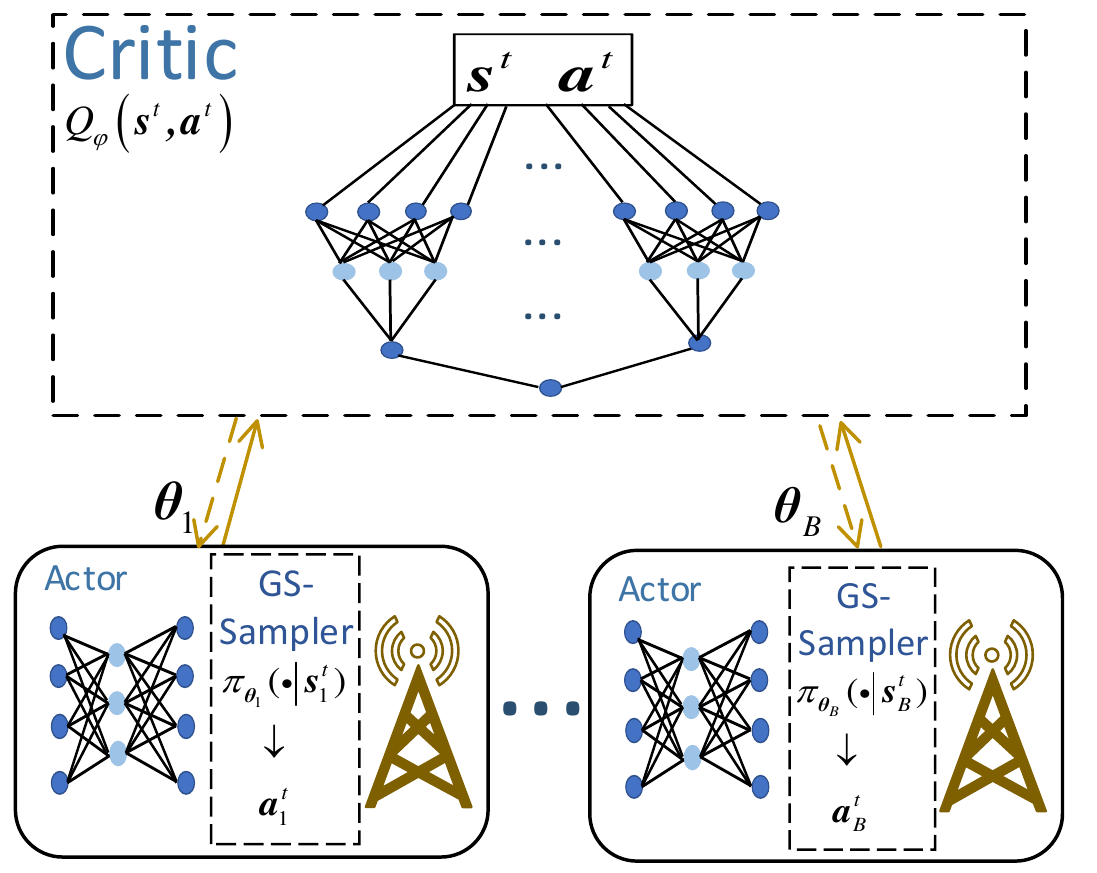}
  \caption{Proposed decentralized control. }
  \label{diagram:madsac_dc}
\end{subfigure}
\caption{Diagrams of algorithm implementations.}
\end{figure}


\begin{Remark}
Lemma \ref{lemma:gs} and Corollary \ref{coro:1} lay the foundations of the proposed approach. As such, the proposed MADSAC can circumvent curse of dimensionality because of the utilization of GS-sampler. Particularly, in the considered problem, the output dimension of the proposed approach is $B(F+1)$, linearly increasing versus the number of ENs or IoT sensors. 
In addition, benefiting from the entropy-regularized formalism, MADSAC has high sample efficiency that is expected to converge faster than DQN and other conventional AC-based algorithms \cite{haarnoja2018soft,sutton2018reinforcement}. 
\end{Remark}

\begin{algorithm}[!t]
  \caption{Centralized Multi-Agent Discrete SAC-Based Cache Update}\label{alg:iotcentralized}
  \begin{algorithmic}[1]
    \State Initialize soft Q-function parameters $\bphi_1, \bphi_2$ 
    \State Initialize policy parameter $\bth$ 
    \State Initialize parameters of target networks $\bphi_i^- \leftarrow \bphi_i$, $i = 1,2 $ 
    \State Initialize {\it RB}
    \For{$ t=0,1,2, \cdots $}
    		\State Observe $\bm s$ and take action $\bm a \sim \pi_{\bth} (\cdot|\bm s)$  
    		\State Observe $\bm s'$ and $r$
    		\State Store $\xi = (\bm s, \bm a, r, \bm s')$ in {\it RB}

    \Procedure{TrainMADSAC}{}
    \State Randomly draw a batch of $N$ experiences as $\Xi_N$
    \For{each $\xi = \left( \bm s, \bm a, r, \bm s' \right) \in \Xi_N $}
     \State Calculate target values: 
     $ y_{\xi} = r + \gamma ( \min_{i = 1,2} Q_{\bphi_i^-} (\bm s', \bm a') - \alpha \log\pi_{\bth}(\bm a'|\bm s')) $, where $\bm a' = GS(\pi_{\bth}(\cdot|\bm s'))$
    \EndFor
    \State Update $\bphi_i$ by taking one-step gradient descent of
    $\frac{1}{N} \sum_{\xi \in \Xi_N} ( Q_{\bphi_i}(\bm s, \bm a) - y_{\xi})^2$, for $i = 1,2 $
   	\State Update $\bth$ by taking one-step gradient descent of 
   	\State $\frac{1}{N} \sum_{\xi \in \Xi_N} \big( \min_{i = 1,2} Q_{\bphi_i}( \bm s, GS(\pi_{\bth}(\cdot|\bm s) )) - \alpha \log \pi_{\bth} (GS(\pi_{\bth}(\cdot|\bm s) )|\bm s ) \big) $ 
   	\State Update temperature $\alpha$ by taking one-step gradient descent of 
   	\State $ - \frac{1}{N} \sum_{\xi \in \Xi_N} \alpha \big(\log \pi_{\bth}(GS(\pi_{\bth} (\cdot |\bm s))| \bm s) - \bar H \big)$
   	\State Update parameters of target networks 
   	 $\bphi_i^{-} \leftarrow \tau \bphi_i + (1- \tau) \bphi_i^{-} $, for $i = 1,2$
    \EndProcedure
    \EndFor 
  \end{algorithmic}
\end{algorithm}

\subsection{Algorithm Implementation}
We present the centralized algorithm implementation for cache update in IoT sensing network, i.e., MADSAC-centralized control (MADSAC-CC). 
A workflow is illustrated in Fig. \ref{diagram:madsac_cc}. 
Particularly, the CP plays a role of the centralized agent and trains a centralized policy for all ENs. Accordingly, each EN should first transfer its local observations (e.g., $\bm s_b^t, \forall b \in \mc B$) to the CP every epoch. After collecting local observations, the CP returns local actions to each EN separately (e.g., $\bm a_b^t, \forall b \in \mc B$). Hereunder, we describe network design and algorithm implementation. 

{\it Network Design:} We adopt a clipped double Q-learning approach, where two separate Q-networks are concurrently trained, and the minimum of the outputs of two networks is the estimate of the soft Q-function. This approach assists to overcome overestimation \cite{fujimoto2018addressing}. We denote the parameterized Q-functions by $Q_{\bphi_1}(\bm s, \bm a)$ and $Q_{\bphi_2} (\bm s, \bm a)$. Note that, the joint action $\bm a$ consists of a number of $B$ one-hot vectors. 
We further maintain a policy network to output a $B(F+1)$-dimensional vector, where each element respectively corresponds to Pr$\{\bm a_b|\bm s^t\}, \forall \bm a_b \in \mc A_b, \forall b \in \mc B$. 

{\it Algorithm Training:} The {\it RB} is supposed to have a finite capacity and store the latest experiences $\xi^t$. At every step, we randomly draw a mini-batch of $N$ samples to approximate loss terms. We present the concrete steps in Algorithm \ref{alg:iotcentralized}. Particularly, the target networks are slowly updated by {\it soft copy} to stabilize training. That is, the constant $\tau$ in step 19 in Algorithm \ref{alg:iotcentralized} is a small positive value. All  parameters, i.e., $\{\bphi_1, \bphi_2, \bth, \alpha\}$, are trained by using stochastic gradient descent with proper learning rates. 

It is worth mentioning that, in the training procedure, the CP needs to concurrently train Q-networks and policy network. When policy network is well-tuned, it can be solely utilized to select actions for all ENs in the evaluation procedure. However, the centralized control requires the CP to first aggregate local observations at ENs and then distribute caching decisions back to ENs at every epoch, which inevitably introduces very high communication overhead. 

\section{Decentralized Multi-Agent Discrete Soft Actor-Critic-Based Caching}
To reduce communication overhead suffered from the centralized control, we develop a decentralized MADSAC-based caching scheme. Particularly, we devise decentralized policies for each agent to locally generate caching decisions, while utilizing the centralized soft-Q function to globally optimize these decentralized policies.

\subsection{Proposed Decentralized Multi-Agent Discrete Soft Actor-Critic Learning}
To generalize the proposed discrete variant of the SAC into the decentralized control, we maintain $B$ parameterized stochastic polices for each agent, i.e., $\pi_{\bth_b}(\cdot |\bm s_b)$, where $\bth_b$ denotes the parameter of the corresponding function approximator, $\forall b \in \mc B$. According to the principle of the SAC, the decentralized MADSAC learning attempts to optimize the following entropy-regularized problem:
\begin{align}
	\max_{\{\bth_b\}} ~ \mathbb{E}_{\{\bm s_b^t\}, \{\bm a_b^t\}} \left[ \sum_{t=0}^{+\infty}(\gamma)^{t}  \bigg(r^{t+1} + \alpha \sum_{b \in \mc B} \mc H(\pi_{\bth_b}(\cdot|\bm s_b^t)) \bigg) \right]. \label{eq:masac} 
\end{align}
Thus, the centralized soft-Q function can be defined as:
\begin{align}
	Q(\{\bm s_b\}, \{\bm a_b\}) = \mathbb{E} \left[ V^t + \sum_{\tau=1}^{+\infty}\sum_{b\in \mc B} (\gamma)^{\tau} \alpha \mc H(\pi_{\bth_b}(\cdot|\bm s_b^{t+\tau}))  \big|\bm s_b^t =  \bm s_b, \bm a_b^t = \bm a_b, \forall b \in \mc B \right].
\end{align}
Accordingly, we further maintain a parameterized soft Q-function $Q_{\bphi} (\{\bm s_b\}, \{\bm a_b\})$, which can be applied to refine local (decentralized) policies. As such, policy parameters $\{\bth_b\}$ can be trained by minimizing the following loss function:
\begin{align}
 	J_\pi (\{\bth_b\}) = \mathbb E_{\{\bm s_b^t\} \sim \Xi} \left [  \mathbb E_{\bm a_b \sim \pi_{\bth_b}} \bigg[ \sum_{b \in \mc B} \alpha \log(\pi_{\bth_b} (\bm a_b| \bm s_b^t))  - Q_{\bphi} \big(\{\bm s_b^t\}, \{\bm a_b\}\big)\bigg]\right]. \label{eq:malossp}
\end{align}
Regarding the update of $\bphi$, we follow the centralized MADSAC step and minimize the following loss:
\begin{align}
	J_Q(\bphi) = \mathbb{E}_{\xi^t \sim \Xi} \left[ \left( Q_{\bphi}\big(\{\bm s_b^t\}, \{\bm a_b^t\}\big) -  (r^{t+1} + \gamma \bar V (\{\bm s_b^{t+1}\}) )\right) ^2 \right], \label{eq:malossq}
 \end{align}
where the target value $\bar V(\{\bm s^{t+1}\})$ is given by:
\begin{align}
	\mathbb E_{\bm a_b \sim \pi_{\bth_b} } \left[ Q_{\bphi^-} (\{\bm s^{t+1}\}, \{\bm a_b\}) - \alpha \sum_{b \in \mc B} \log \pi_{\bth_b} (\bm a_b|\bm s_b^{t+1}) \right], \label{eq:matargetq}
\end{align}
and $Q_{\bphi^-}(\{\bm s_b\}, \{\bm a_b\})$ is the target network and its parameter $\bphi^-$ is slowly updated every epoch. Last, temperature parameter $\alpha$ can be updated similarly to \eqref{eq:losstem}. Likewise, the expectation in \eqref{eq:malossp}-\eqref{eq:matargetq} can be approximated by using samples drawn from the {\it RB} or produced by the GS-sampler. We present the algorithm implementation of MADSAC-decentralzed control (MADSAC-DC) in the ensuing subsection.

\begin{algorithm}[!t]
  \caption{Decentralized Multi-Agent Discrete SAC-Based Cache Update}\label{alg:iotdecentralized}
  \begin{algorithmic}[1]
    \State Initialize centralized soft Q-function parameters $\bphi_1, \bphi_2$ 
    \State Initialize policy parameters $\{\bth_b\}$ 
    \State Initialize parameters of target networks $\bphi_i^- \leftarrow \bphi_i$, $i = 1,2 $ 
    \State Initialize {\it RB}
    \For{$ t=0,1,2, \cdots $}
    		\State Observe $\bm s_b$ and take action $\bm a_b \sim \pi_{\bth_b} (\cdot|\bm s_b)$ for $b \in \mc B$ 
    		\State Observe $\bm s_b'$ for $b \in \mc B$ and reward $r$
    		\State Store $\xi = \big(\{\bm s_b\}, \{\bm a_b\}, r, \{\bm s_b'\}\big)$ in {\it RB}

    \Procedure{TrainMADSAC}{}
    \State Randomly sample a batch of $N$ experiences as $\Xi_N$
    \State Update $\{\bth_b\}$ by taking one-step gradient descent of $J_{\pi}(\{\bth_b\})$ 
    \State Update $\bphi$ by taking one-step gradient descent of $J_Q(\bphi)$
    \State Update $\alpha$ by taking one-step gradient descent of $J(\alpha)$ 
   	\State Update parameters of target networks
   	 $\bphi_i^{-} \leftarrow \tau \bphi_i + (1- \tau) \bphi_i^{-} $, for $i = 1,2$
    \EndProcedure
    \EndFor 
  \end{algorithmic}
\end{algorithm}
\subsection{Algorithm Implementation}
As illustrated in Fig. \ref{diagram:madsac_dc}, the proposed MADSAC-DC operates as follows: in the training procedure, each EN needs to upload local observations to the CP, which then globally optimizes policy parameters. Again, we adopt the technique of the clipped double Q-learning, and maintain two Q-network to combat overestimation. The training steps are summarized in Algorithm \ref{alg:iotdecentralized}. Finally, the CP should deliver the well-trained policy parameters to each EN, which thereafter is able to produce local actions based on its observation only.

\begin{Remark}
	Clearly, in the training procedure, the decentralized control suffers the same amount of communication overhead as that of MADSAC-CC, i.e., $\mc O(B^2F + B)$, due to information aggregation of user requests (e.g., $\{N_{f,b}^t\}$) and distribution of local actions (e.g., $\{\bm a_b^t\}$). When it goes into the evaluation procedure, each EN can independently make decisions with null communication overhead while the communication overhead of MADSAC-CC remains to be $\mc O(B^2F + B)$. 
\end{Remark}
\section{Performance Evaluation}
In this section, we evaluate the performance of the proposed algorithms under various kinds of scenarios. Unless otherwise stated, the default setting is as follows: three ENs are considered in an IoT sensing network; each EN has a communication range of 100 m and can coordinate 10 randomly distributed  sensors; the storage of every content item is randomly generated within $[0.05, 0.1]$ GB. The range of the AoI is $[1, 50]$. Concerning communications between IoT sensors and ENs, the transmission power is 20 dBm at each sensor; the path loss is $- (148.1 + 37.6 \log_{10} d)$ dB with $d$ being the distance in km; the channel bandwidth is 10 MHz; the antenna gain is 10 dBi; the log-normal shadowing parameter is 8 dB; and the received SNR threshold is specified by $\eta_{th} = 10$ dB. Furthermore, the space-time popularity dynamics of user requests are modeled as follows: user requests at distinct ENs exhibit individual content popularity distributions; at each EN, there are at most 100 users making requests according to a class of Zipf distributions  \cite{sadeghi2018optimal}, namely, 
$
	p_{f,b} = \zeta_{f,b}^{-\upsilon_{b}}/\textstyle \sum_{f' \in \mc F} \zeta_{f',b}^{-\upsilon_{b}}, \forall f, b,  
$
where $\upsilon_{b}$ denotes the skewness factor that is selected from $\{0.5, 1, 1.5, 2\}$, and $\{\zeta_{f,b}\}$ denote rank orders of content items that are dynamically evolving by following certain transition probability matrix  \cite{sadeghi2018optimal}. In addition, we consider $\omega_1 = \omega_2  = 1$.

In the subsequent subsections, we consider the following algorithms for comparison:
\begin{itemize}
	\item {\bf DQN:} This algorithm is widely used in prior works (e.g., \cite{abd2020reinforcement,xu2020aoi}) to handle update at a single EN. It is expected to obtain near-optimal results in small-scale settings, which attempts to validate the effectiveness of the proposed algorithm. 
	\item {\bf AC:}  We consider a popular AC-based algorithm in \cite{Lillicrap2015ContinuousCW}. To apply it in the discrete decision-making, we again adopt GS-sampler and recast the output into low dimension, similar to the proposed algorithms. This approach is utilized to evaluate the potentials of entropy regularization. 
\end{itemize}
The above-mentioned algorithms are assumed to operate in the centralized manner, similar to MADSAC-CC. 

\subsection{Learning Curves of the Proposed Algorithms}
 We design either Q-networks or policy networks as five-layer neural networks, which consist of an input layer, three hidden layer, and an output layer, respectively. 
 To stabilize training, the {\it polynomial learning rate} policy is adopted to train networks (readers are referred to \cite{mishra2019polynomial} in detail.). We summarize key parameters of algorithm implementation in Table \ref{table:implementation}. To ensure fairness, all algorithms are implemented by the same configuration.

We illustrate the learning curves of all the algorithms in Fig. \ref{fig:lcb1} and \ref{fig:lcb3}, respectively. Specifically, the algorithm performance is presented in terms of average reward (shown by learning curves) and standard deviation (shown by shaded areas). All of the results are obtained by applying the moving average, i.e., $\sum_{\tau = t-T+1}^{t} r^{\tau}$, where $T = 5000$. In Fig. \ref{fig:lcb1}, we consider a single agent setting (i.e., EN) to validate the effectiveness of the proposed discrete variant of SAC. Clearly, the proposed MADSAC converges very fast and achieves comparable final results as that of DQN, whilst AC takes a much longer while to converge. Moreover, we can observe some sudden drops in the curve of AC, whereas the curve of the proposed one is generally flat. The observation indicates that the proposed approach is able to learn more stably.  
We further consider a multi-agent setting with three ENs in Fig. \ref{fig:lcb3}. Evidently, the proposed algorithms outperform DQN and AC. At the initial stage, the rewards of MADSAC-CC and MADSAC-DC increase faster than AC, and shortly converge to almost the same level. We can observe notable gaps between final results of the proposed algorithms and AC. This finding implies that the entropy-regularized objective in the proposed approach is able to circumvent premature convergence somehow. However, DQN fails to make meaningful progress in the multi-agent setting. The reason for this is that an explosive action space makes DQN difficult to estimate the values of the Q-function. 
These results confirm the remarkable performance of the proposed algorithms in terms of convergence speed and final performance, favorably to the results obtained in \cite{haarnoja2018soft}. 

\begin{table}[h!]
\centering
\caption{Parameters for algorithm implementations}
\begin{tabular}{c c}
 \hline
  Parameters & Value\\
  \hline
  Number of neurons in each hidden layer  &128  \\
  Optimizer & {\it Adam}\\
  Initial learning rate for Q-networks & 0.01\\
  Initial learning rate for policy networks & 0.001\\
  Power factor for decreasing learning rates & 0.9\\
  Memory capacity of {\it RB} & 5000\\
  Mini-batch size & 100\\
  Step size for updating target networks & 0.001\\
  Discount factor & 0.99\\
 \hline
\end{tabular}
\label{table:implementation}
\end{table}

\begin{figure}
\begin{subfigure}[t]{.495\linewidth}
  \centering
  \includegraphics[scale=0.42]{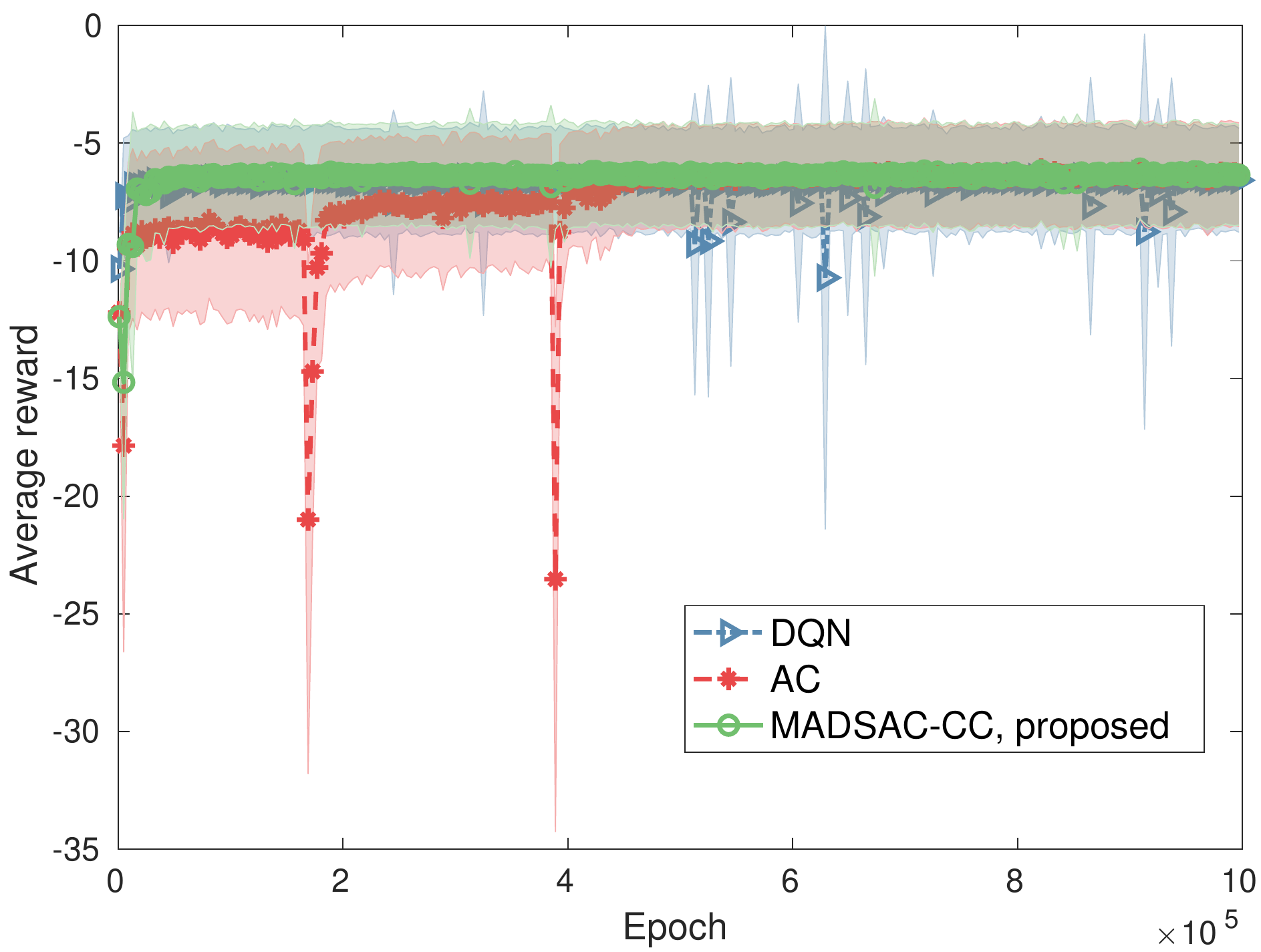}
  \caption{Single agent}
  \label{fig:lcb1}
\end{subfigure}
\begin{subfigure}[t]{.495\linewidth}
  \centering
  \includegraphics[scale=0.42]{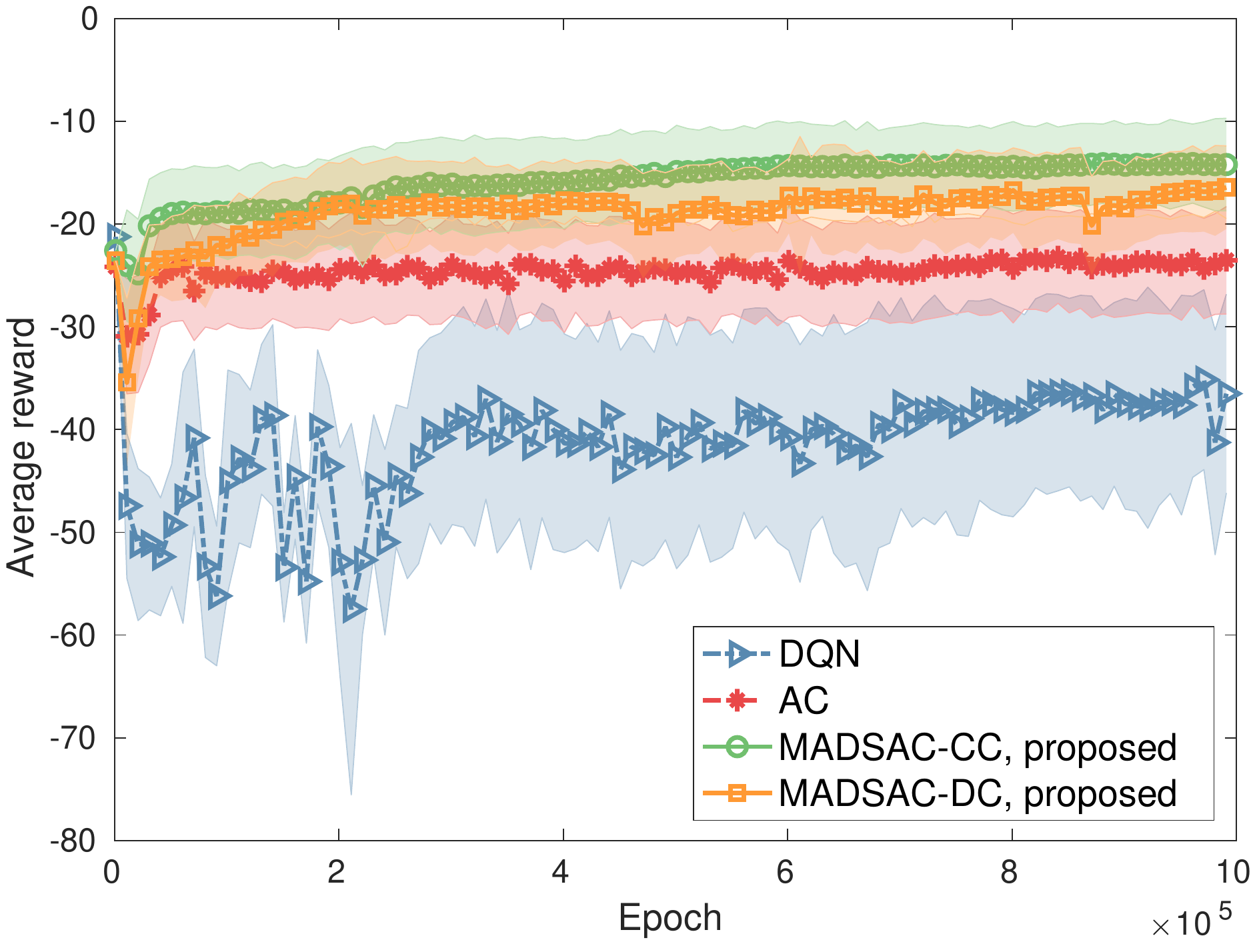}
  \caption{Multiple agents. }
  \label{fig:lcb3}
\end{subfigure}
\caption{Learning curves.}
\end{figure}

\subsection{Scalability}
In this subsection, we investigate the impacts of system parameters and study the scalability of the proposed algorithms. To what follows, we use the average weighted cost as the performance criterion, i.e., defined in \eqref{eq:cost}. All of the results are obtained by averaging over 10000 epochs after DNNs are well-tuned. 

Particularly, we first vary the number of agents (e.g., ENs) and plot the results in Fig. \ref{fig:en}. As can be observed, the proposed algorithms achieve almost comparable results, when the number of agents is no larger than five. This is because the utilization of the centralized critic in  MADSAC-DC can effectively criticize the decisions of each local agent. Thus, it assists to agent collaboration and leads to globally optimized decisions. As the number of agents becomes large, the performance of the proposed decentralized algorithm gradually degrades due to the restriction of local observations. Nevertheless, over the entire horizontal axis, the decentralized algorithm always obtains better results than the AC-based scheme that operates in the centralized manner. When nine ENs are deployed, MADSAC-CC and MADSAC-DC can reduce 42.44\% and 27.51\% of the average weighted cost, respectively, in comparison to AC. It worth mentioning that, when five or more ENs are considered, it is not practical to implement DQN due to the extremely large number of actions. Moreover, we illustrate the results of each considered metric (achieved by MADSAC-CC) in Fig. \ref{fig:en2}. As can be seen, the average AoI becomes larger where more ENs are available. The reason is that more content items are involved with the growing number of ENs. Interestingly, transmission energy consumption and fronthaul traffic loads witness increase trends when the number of agents are less than five but degrades a little bit afterwards. Our conjecture is that the decrease in energetic cost and traffic loads is as a result of worsening data freshness. 

To further investigate the scalability, we carry out experiments by changing the number of sensors within the communication range of each EN. As shown in Fig. \ref{fig:file}, the weighted cost, achieved by the decentralized design, is quite close to that of the centralized one. This observation further corroborates the remarkable performance of the decentralized control. As anticipated, MADSAC-CC and MADSAC-DC achieve much lower average weighted costs than DQN. Specifically, when $20$ sensors are deployed at the coverage of each EN,  the proposed DRL schemes are able to reduce the weighted cost by 54.23\%, 50.94\%, respectively, compared with the DQN-based scheme. Notably, when 25 sensors are considered, the resulting number of discrete actions is 17576, making DQN implementation impossible. Similarly, the proposed DRL schemes outperform the AC-based scheme over the entire horizontal axis.
However, we should mention that the utilization of entropy-regularization in the proposed algorithms does not alway have a remarkable advantage over the conventional RL. The AC based scheme sometimes achieves comparable performance as the proposed ones, e.g., in the case of 15 sensors.
We further show the results of each performance criterion (achieved by MADSAC-CC) in Fig. \ref{fig:files}. The average AoI gradually becomes large because of the enlargement of content  catalog. Average transmission energy and traffic loads exhibit similar results to what has been found in Fig. \ref{fig:en2}.
\begin{figure*}
\begin{subfigure}[t]{.495\linewidth}
  \centering
  \includegraphics[scale=0.43]{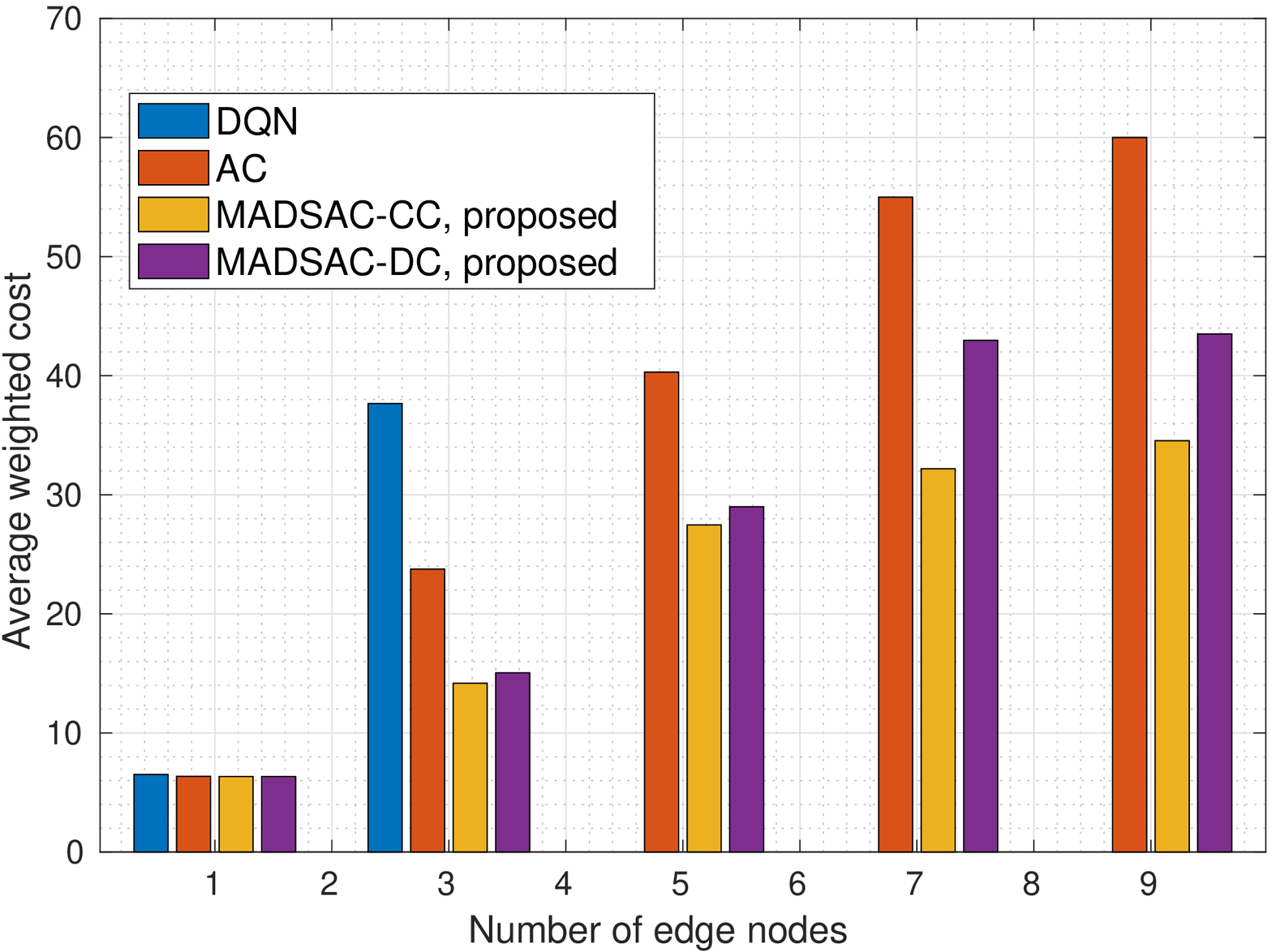}
  \caption{Average weighted cost. }
  \label{fig:en}
\end{subfigure}
\begin{subfigure}[t]{.495\linewidth}
  \centering
  \includegraphics[scale=0.42]{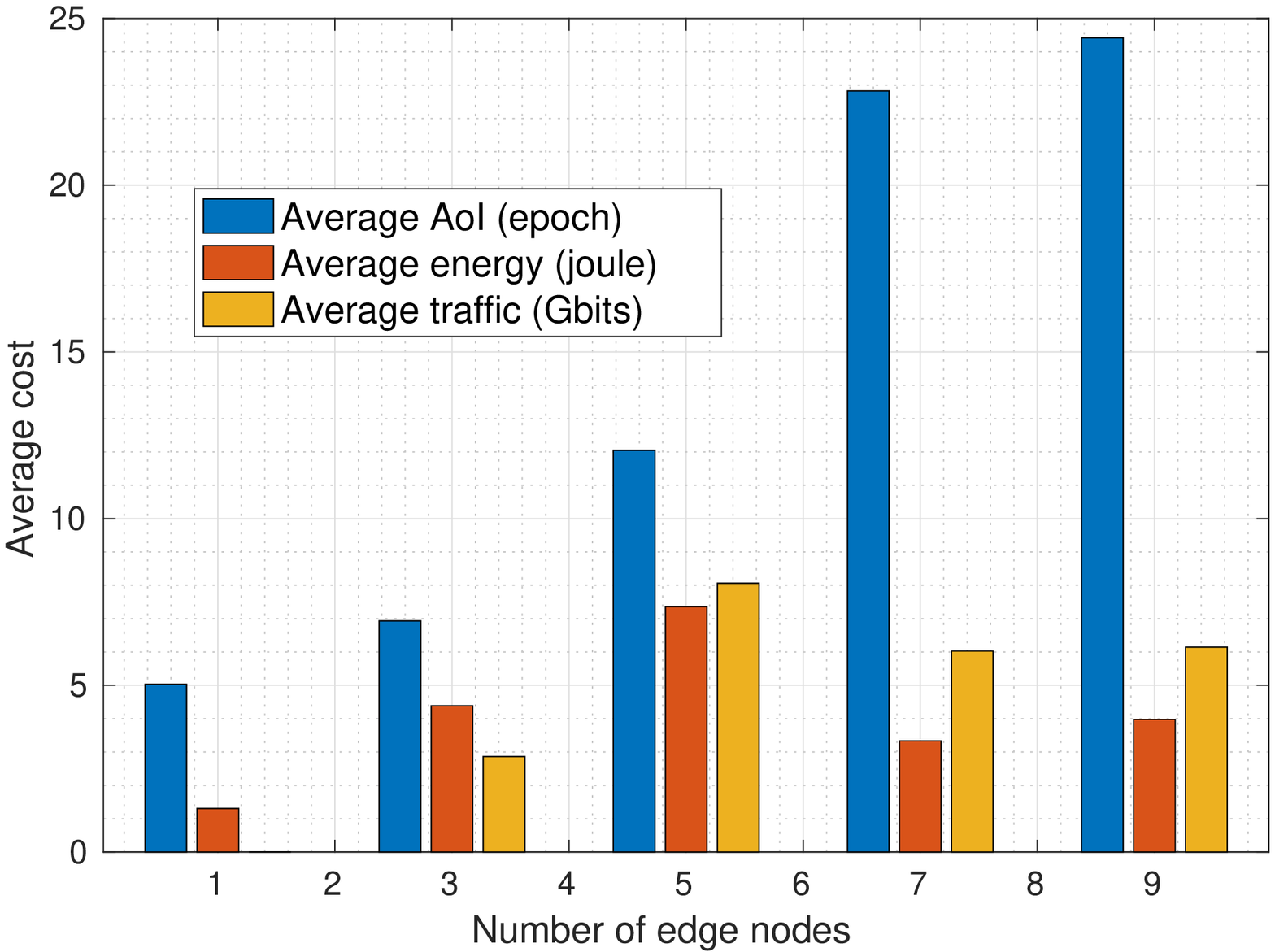}
  \caption{Individual performance criterion. }
  \label{fig:en2}
\end{subfigure}
  \caption{Impacts of the number of ENs. }
\end{figure*}

\begin{figure*}
\begin{subfigure}[t]{.495\linewidth}
  \centering
  \includegraphics[scale=0.42]{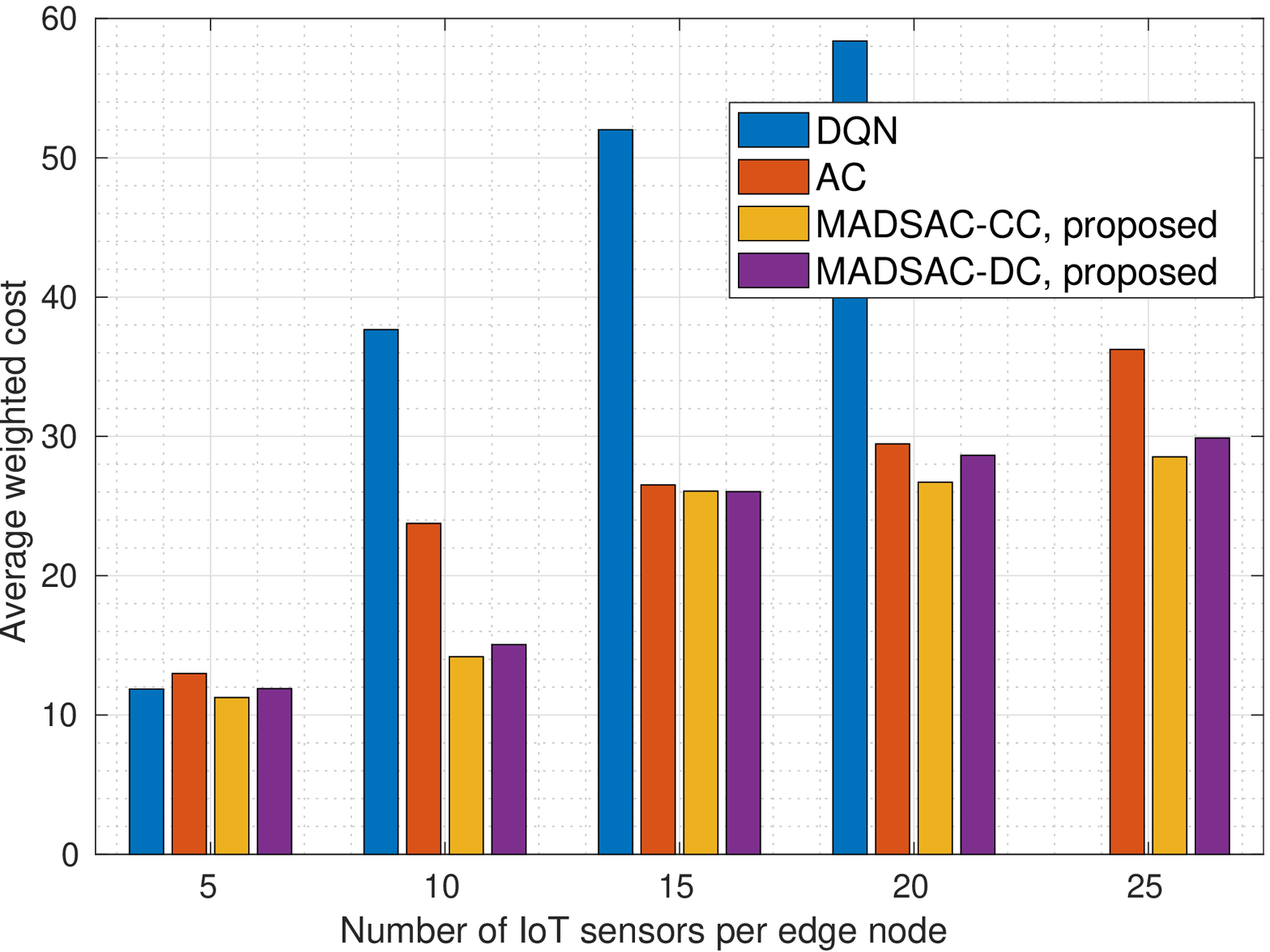}
  \caption{Average weighted cost. }
  \label{fig:file}
\end{subfigure}
\begin{subfigure}[t]{.495\linewidth}
  \centering
  \includegraphics[scale=0.42]{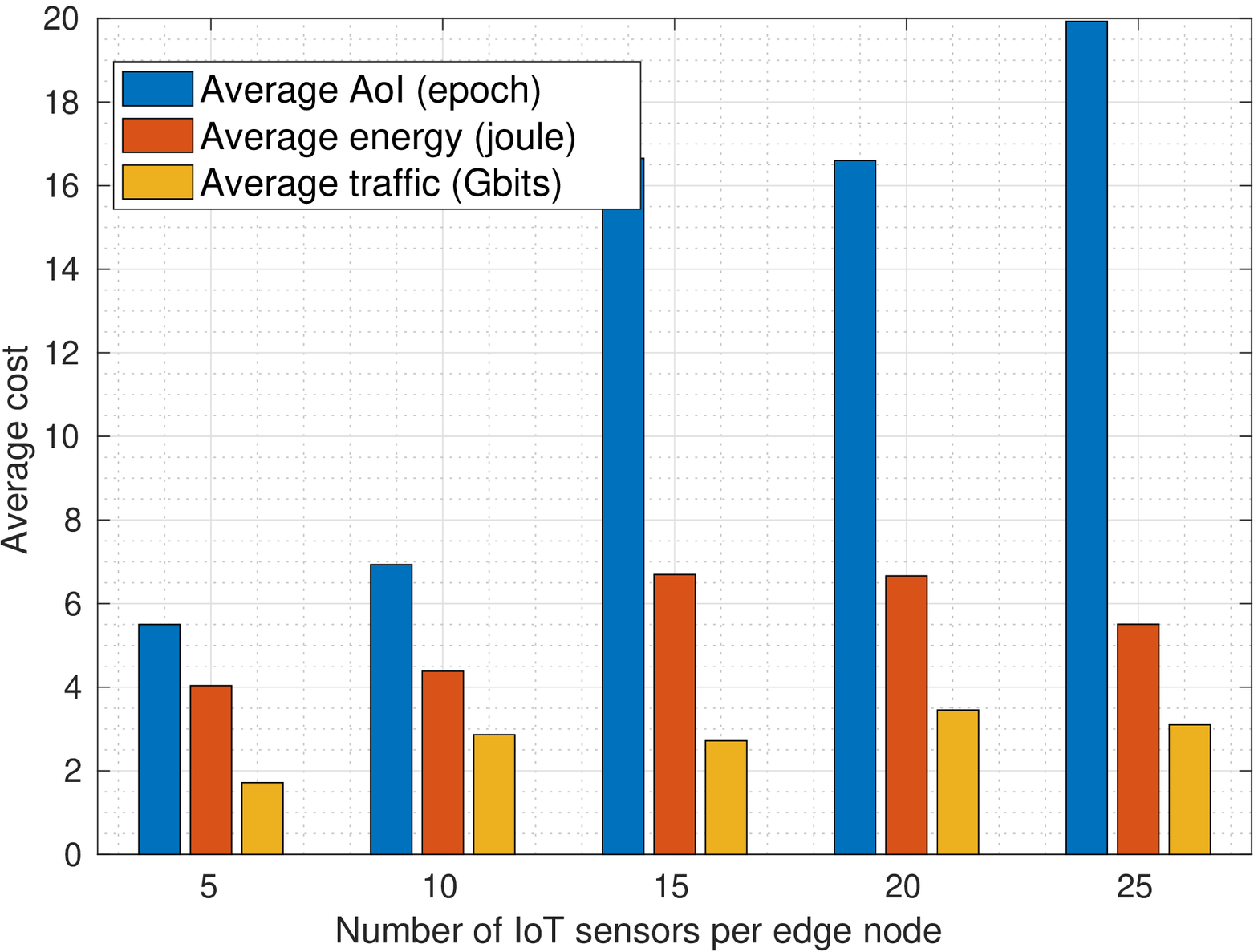}
  \caption{Individual performance criterion.}
    \label{fig:files}
\end{subfigure}
\caption{Impacts of the number of IoT sensors at each EN.}
\end{figure*}

The above-mentioned simulation results confirm the superiority of the proposed discrete variant of SAC, and the generalization of the decentralized approach. In the ensuing section, we only implement MADSAC-CC and focus on the tradeoff among the considered performance criteria. 
\begin{figure}
\begin{minipage}[t]{.495\linewidth}
  \centering
  \includegraphics[scale=0.42]{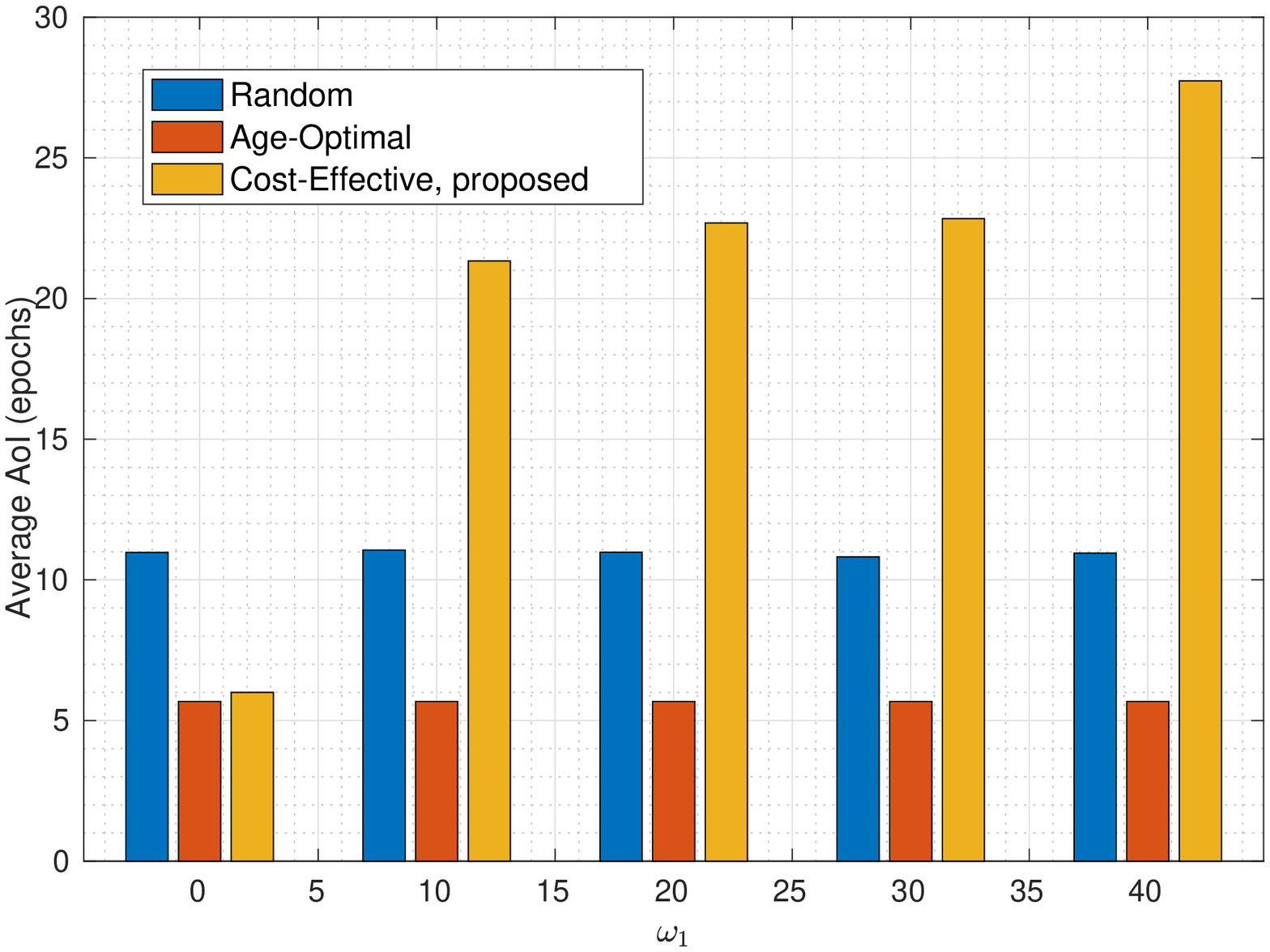}
  \caption{Average AoI versus $\omega_1$.}
  \label{fig:eta1aoi}
\end{minipage}
\begin{minipage}[t]{.495\linewidth}
  \centering
  \includegraphics[scale=0.42]{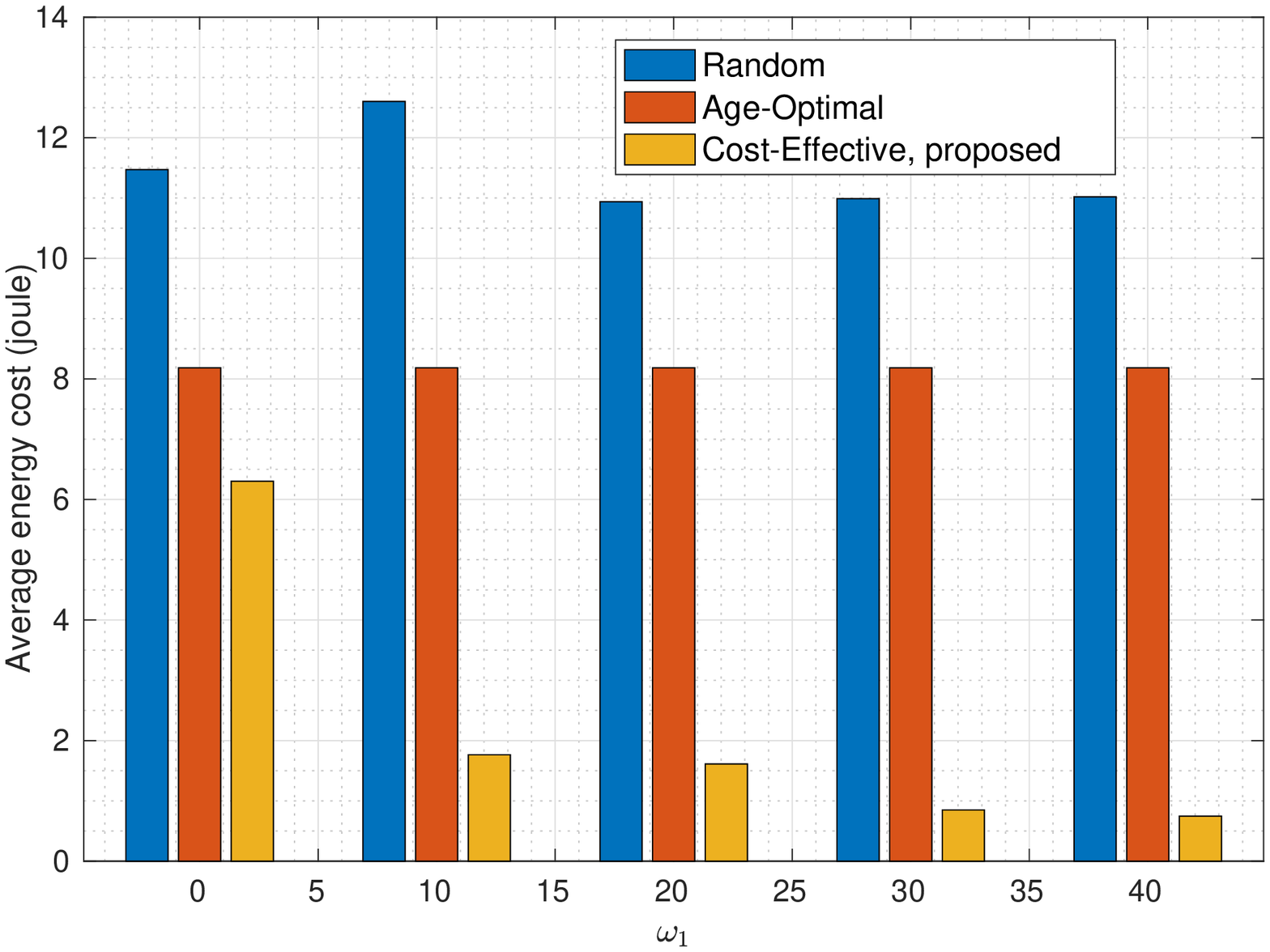}
  \caption{Average transmission energy consumption versus $\omega_1$. }
  \label{fig:eta1egy}
\end{minipage}
\end{figure}

\begin{figure}
\begin{minipage}[t]{.495\linewidth}
  \centering
  \includegraphics[scale=0.42]{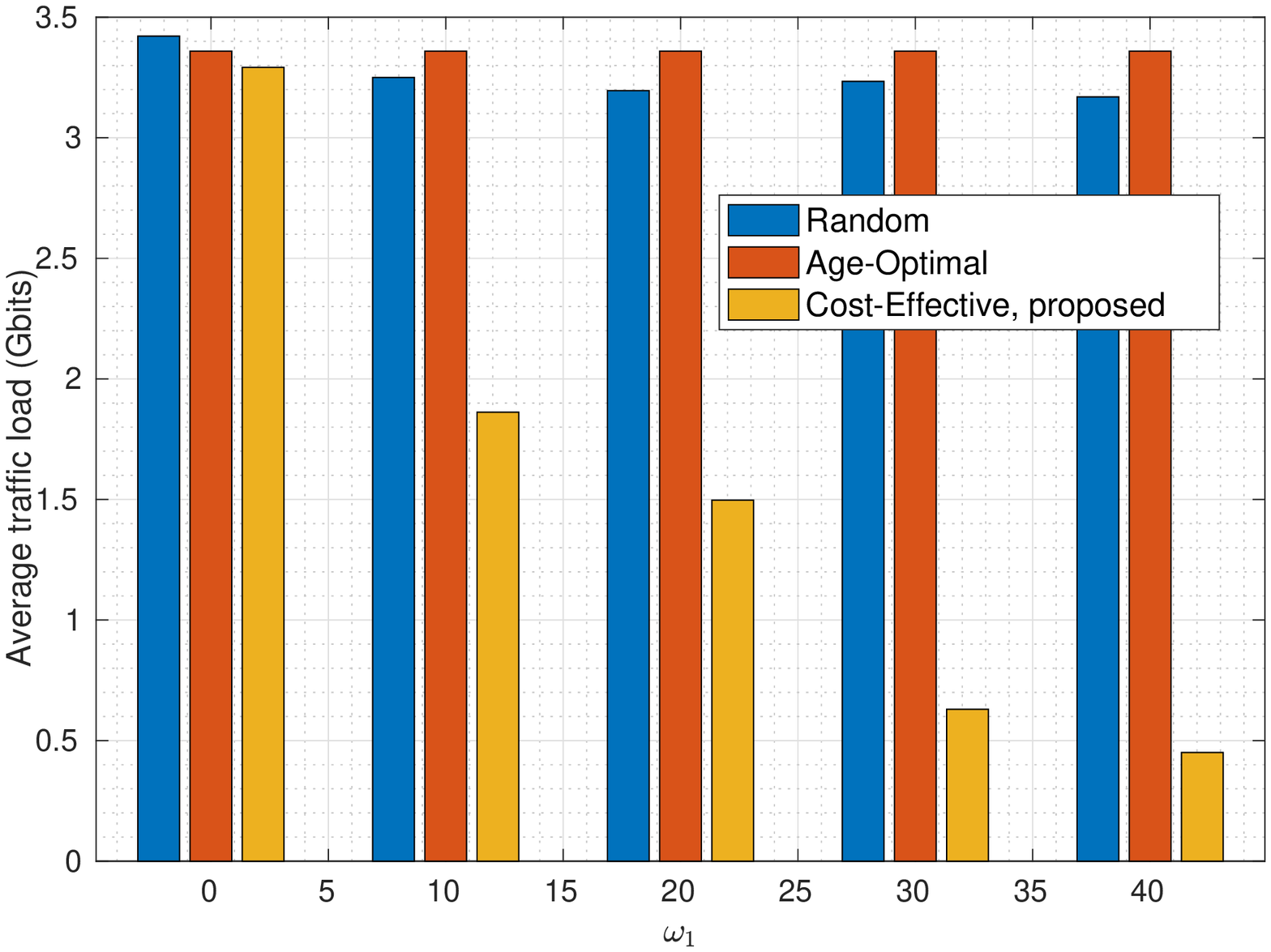}
  \caption{Average fronthaul traffic loads versus $\omega_1$. }
  \label{fig:eta1load}
\end{minipage}
\begin{minipage}[t]{.495\linewidth}
  \centering
  \includegraphics[scale=0.42]{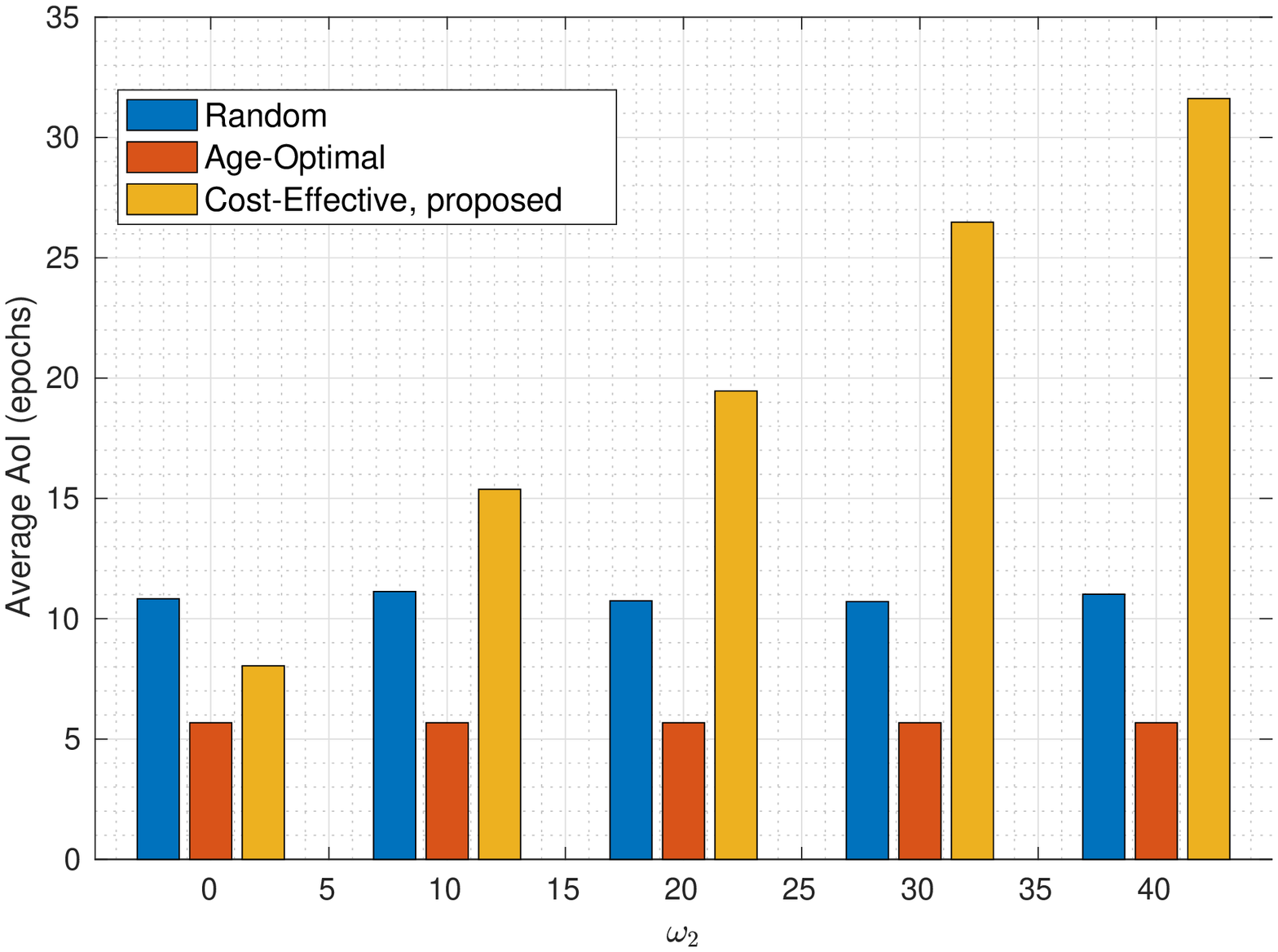}
  \caption{Average AoI versus $\omega_2$. }
  \label{fig:eta2aoi}
\end{minipage}
\end{figure}

\begin{figure}
\begin{minipage}[t]{.495\linewidth}
  \centering
  \includegraphics[scale=0.42]{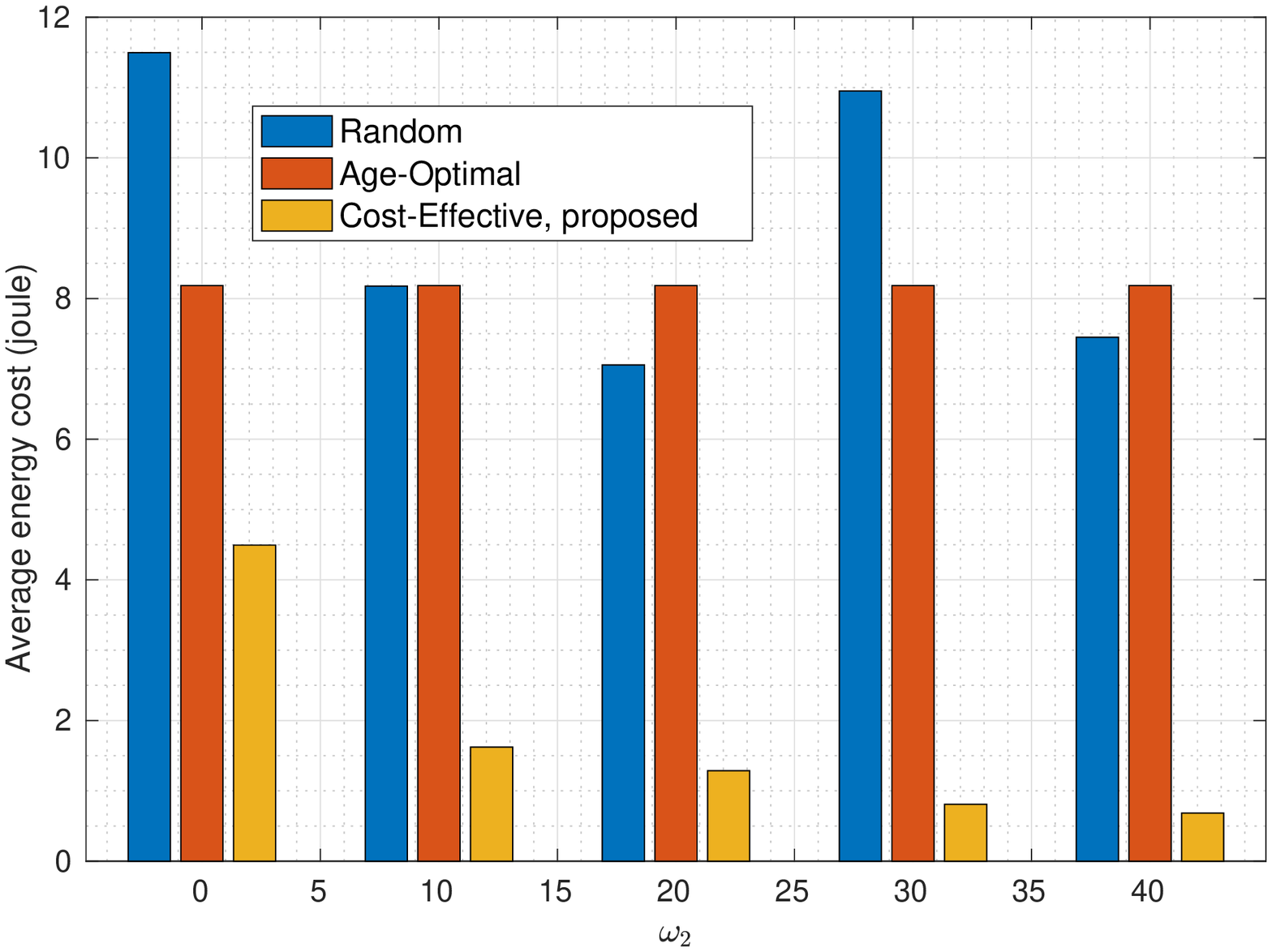}
  \caption{Average transmission energy consumption versus $\omega_2$. }
  \label{fig:eat2egy}
\end{minipage}
\begin{minipage}[t]{.495\linewidth}
  \centering
  \includegraphics[scale=0.42]{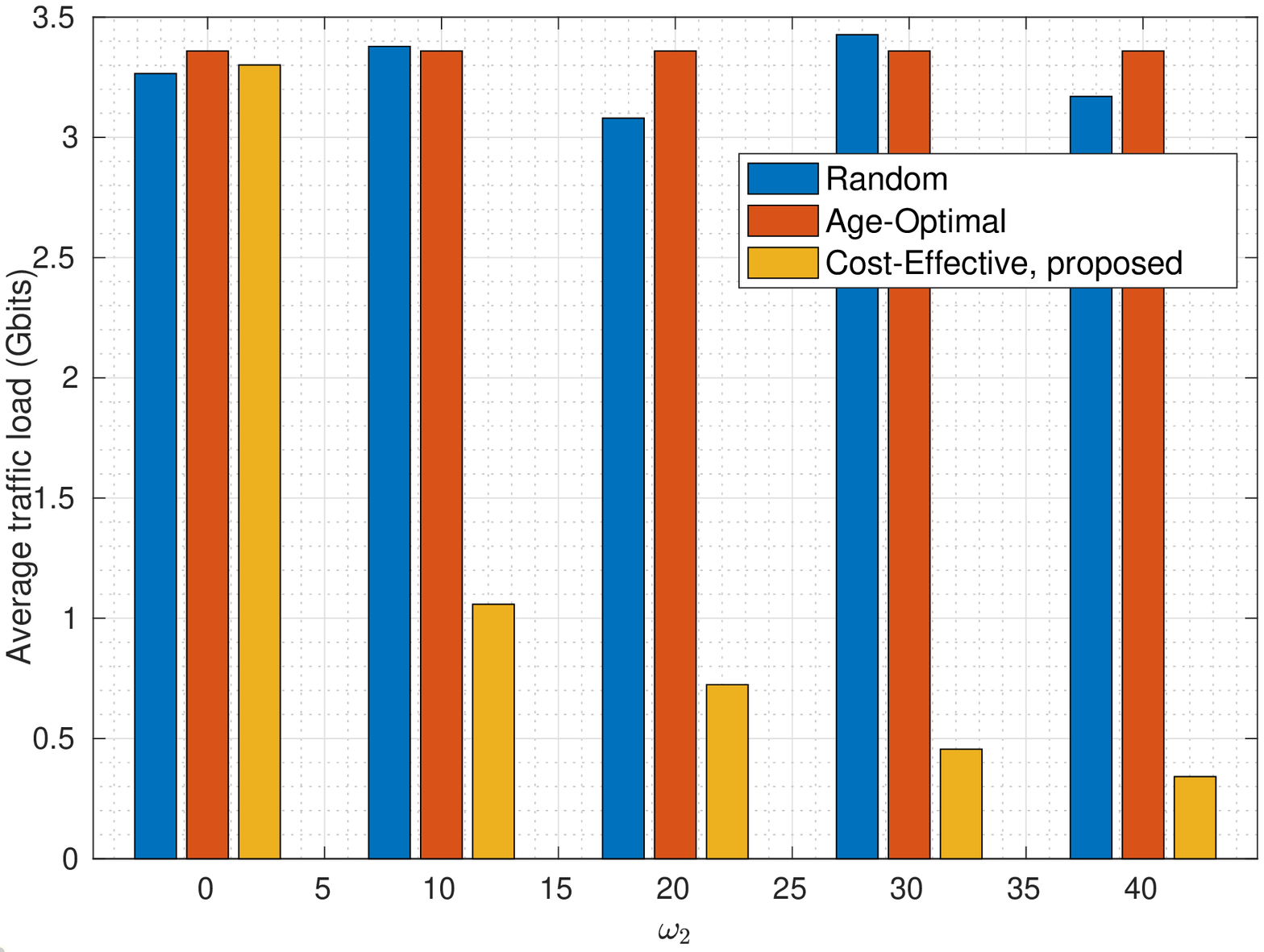}
  \caption{Average fronthaul traffic loads versus $\omega_2$. }
  \label{fig:eat2load}
\end{minipage}
\end{figure}

\subsection{Tradeoff among AoI \& Energy Consumption \& Traffic Loads}
In this subsection, we investigate the tradeoff among average AoI, transmission energy consumption and fronthaul traffic loads. 
To benchmark how energy consumption and traffic load consideration compromises the performance of data freshness, we consider two schemes as follows. i) {\it Age-Optimal Scheme:} We only optimize the average AoI without incorporating cache update costs; ii) {\it Random Scheme:} at each epoch, we randomly update one content item at each EN without being aware of the tradeoff among the considered performance criteria. The proposed scheme is referred to as {\it cost-effective scheme}.

We first carry out experiments by varying $\omega_1$ and illustrate the results of average AoI, transmission energy consumption, and fronthaul traffic loads in Fig. \ref{fig:eta1aoi} - \ref{fig:eta1load}, respectively. Besides, $\omega_2$ is fixed as the default value (i.e., 1). 
Evidently, when we enlarge the weight for transmission energy (i.e., $\omega_1$), the average AoI continuously grows high whereas transmission energy drops off quickly. This finding implies that the cost-effective scheme attempts to reduce the frequency of content update. For this reason, it leads to the reduction of traffic loads simultaneously, although the associated weight is fixed as a constant.  

We then conduct simulations by changing $\omega_2$ and fixing $\omega_1 = 1$. The results are depicted in Fig. \ref{fig:eta2aoi} - Fig. \ref{fig:eat2load}. It can be observed that, with the increment of $\omega_2$,  fronthaul traffic loads decrease gradually while the average AoI becomes increasingly large. Similarly, we conjecture that the update of the content becomes less frequent, which somehow results in more outdated content items cached at ENs. 
In addition, it should be noted that that fronthaul traffic loads, achieved by enlarging $\omega_2$ (shown in Fig. \ref{fig:eat2load}), decreases faster than that in Fig. \ref{fig:eta1load}.  For instance, when increasing $\omega_2$ up to 10, there is a 67.92\% reduction of traffic loads, which is much lager than 42.16\% achieved by enlarging $\omega_1$. This is because, when we enlarge $\omega_2$, the cost-effective scheme is likely to reduce the update frequencies of content items having large storage size. A similar conclusion can be drawn towards the degradation of energy consumption by tuning $\omega_1$. 
Finally, when we set $\omega_1$ or $\omega_2$ to be larger than 10, the reduction of energy consumption and fronthaul traffic loads are quite limited in comparison to age-optimal scheme or random scheme; it however leads to much larger AoI compared with baselines. Hence, $\omega_1$ and $\omega_2$ should be reasonably adjusted in order to balance the average AoI and update costs in practice.

\section{Conclusion}
In this paper, we have developed a multi-agent reinforcement learning framework for cache update in IoT sensing networks.
The objective of this framework is to minimize the weighted average age of information plus energy cost as well as fronthaul traffic loads. We have derived a characterization of energy consumption for content delivery. To cope with the discrete multi-agent decision-making, we have proposed a novel reinforcement learning approach with low space complexity. Simulation results have indicated that the proposed algorithms significantly outperform deep Q-network and traditional actor-critic approaches as the number of edge nodes or sensors increases; and the proposed decentralized caching scheme obtains satisfactory performance compared with the centralized one. The developed approach also has great potential to be applied in many other multi-agent discrete decision-making tasks. 

\appendix
\subsection{\textit{Proof of Proposition \ref{prop:egy}}}   
\label{appen:A}

To begin with, $\kappa_f^2$ follows an exponential distribution, i.e., $\frac{1}{2}\exp(-\frac{\kappa_f^2}{2})$. Therefore, the distribution of the received SNR $\eta_f$ is given by $\mathbb{P} (\eta_f) = \frac{1}{2\beta_f} \exp(-\frac{\eta_f}{2\beta_f})$. Recall that content transmissions are effective only when the received SNR exceeds $\eta_{th}$. As a result, the expected throughput $\bar R_f$ can be calculated as follows: 
\begin{align*}
	\bar R_f &= \int_{\eta_{th}}^{\infty} B \log_2(1+\eta_f) \mathbb{P}(\eta_f) ~d \eta_f\\
	& = \frac{B}{\log 2} \int_{\eta_{th}}^{\infty} \log(1+\eta_f) \frac{1}{2\beta_f} \exp\left(-\frac{\eta_f}{2\beta_f}\right) ~d\eta_f\\
	& = - \frac{B}{\log 2} \int_{\eta_{th}}^{\infty} \log (1+\eta_f) ~ d\exp\left(-\frac{\eta_f}{2\beta_f}\right)\\
	& = - \frac{B}{\log 2} \left[ -\log(1+\eta_f) \exp\left(- \frac{\eta_f}{2\beta_f}\right)\big|_{\eta_{th}}^{\infty}  + \int_{\eta_{th}}^{\infty} \exp{\left(-\frac{\eta_f}{2\beta_f}\right)} \frac{1}{1+\eta_f} ~d\eta_f \right]\\
	& = R_{th} \exp \left(-\frac{\eta_f}{2\beta_f}\right) + \frac{B}{\log 2} \int_{\eta_{th} + 1}^{\infty} \exp\left( \frac{1 - \eta_f } {2\beta_f}\right) \frac{1}{\eta_f} ~ d \eta_f\\
	& = R_{th} \exp{\left(- \frac{\eta_f}{2\beta_f}\right)} + \frac{B}{\log 2} \exp \left({\frac{1}{2\beta_f}}\right) \rho_f (\eta_{th} +1),
\end{align*}
where function $\rho_f(\cdot)$ is defined by \eqref{eq:rpo}. Thus, given the content size $s_f$ and transmission power $P_f$, the average energy consumption can be given by $\bar{E}_f = P_fs_f/\bar{R}_f$. This completes the proof. 

\subsection{\textit{Proof of Lemma  \ref{lemma:gs}}}
\label{appen:B}

For notational convenience, we denote $ \alpha_i = \mu_{\bth}(i|\bm s)$, and $ \varpi_i = g_{i,b} + \log  \alpha_i,  \forall i \in \mc A_b, b \in \mc B$. Then, it follows that $\hat{\bm a}_b = \arg\max_{i \in \mc A_b}  \varpi_i $. We calculate the following probability:
\begin{align*}
	\pr \{ \hat{\bm a}_b = i |\bm s\} & =  \pr \{ \varpi_i \geq \varpi_j, \forall j \neq i\} \\
	& = \int_{-\infty}^{\infty}  \Pi_{j \neq i} \{\varpi_i \geq \varpi_j| \varpi_i\}  \pr\{\varpi_i \} d\varpi_i \\
    & =\int_{-\infty}^{\infty}  \Pi_{j \neq i}  \exp \left({-\exp \left({-\varpi_i + \log \alpha_j}\right)}\right)  \exp \left({- ( \varpi_i - \log \alpha_j + \exp \left({-(\varpi_i  - \log (\alpha_i))}\right)  )}\right)   d\varpi_i\\
    & = \int_{-\infty}^{\infty} \exp \left({-\sum_{j \neq i} \alpha_j \exp({-\varpi_i}) } \right) \alpha_i \exp \left({- (\varpi_i + \alpha_i \exp ({-\varpi_i})}\right) d \varpi_i\\
    & \mathop{=}\limits^{(a)} \int_{-\infty}^{\infty} \alpha_i  \exp \left({- \varpi_i - \exp({-\varpi_i})}\right) d \varpi_i\\
    & =  \alpha_i
\end{align*}
for $\forall i \in \mc A_b, b \in \mc B$, where step $(a)$ is as a result of $\sum_{i \in \mc A_b} \mu_{\bth} (i|\bm s) = 1$. Then, we conclude that Pr$\{\hat{\bm a}|\bm s\} = \Pi_{b \in \mc B} ~ \text{Pr} \{\hat{\bm a}_b|\bm s\} = \Pi_{b \in \mc B} ~\mu_{\bth} (\hat{\bm a}_b|\bm s)$, which completes the proof. 

\bibliographystyle{IEEEtran}
\bibliography{references}
\end{document}